\documentclass[11pt]{article}
\usepackage{amsmath,amsthm,amssymb,hyperref, color,fullpage}
\newcommand{\remove}[1]{}

\newtheorem{thm}{Theorem}[section]
\newtheorem{claim}[thm]{Claim}
\newtheorem{lem}[thm]{Lemma}
\newtheorem{define}[thm]{Definition}

\newtheorem{obs}[thm]{Observation}

\newtheorem{THM}{Theorem}

\newtheorem{prop}[thm]{Proposition}

\def\neg{\text{neg}}

\def\F{{\mathbb{F}}}

\def\R{{\mathbb{R}}}

\def\cB{{\cal B}}
\def\cK{{\cal K}}

\def\cL{{\cal L}}

\def\E{{\mathbb E}}

\newcommand{\ip}[2]{\langle #1,#2 \rangle}

\def\_{\,\,\,\,\,}

\def\tr{\textsf{tr}}

\def\poly{\textsf{poly}}

\def\span{\textsf{span}}
\def\Cor{\textsf{Cor}}

\newcommand{\barconst}{\lambda}

\begin{document}

\title{Breaking the quadratic barrier for 3-LCC's over the Reals}
%\author{Zeev Dvir \and Shubhangi Saraf \and Avi Wigderson}

 \author{Zeev Dvir\thanks{Department of Computer Science and Department of Mathematics, Princeton University.
 Email: \texttt{zeev.dvir@gmail.com}. Research  supported by NSF grants CCF-1217216 and CCF-0832797.} \and
 Shubhangi Saraf\thanks{Department of Computer Science and Department of Mathematics, Rutgers University.
 Email: \texttt{shubhangi.saraf@gmail.com}.} \and
 Avi Wigderson\thanks{School of Mathematics, Institute for Advanced Study.
 Email: \texttt{avi@ias.edu}.}}

\date{}
\maketitle

\begin{abstract}

We prove that 3-query linear locally correctable codes over the Reals of dimension $d$ require block length $n>d^{2+\barconst}$ for some fixed, positive $\barconst >0$. Geometrically, this means that if $n$ vectors in $\R^d$ are such that each vector is spanned by a linear number of disjoint triples of others, then it must be that $n > d^{2+\barconst}$. This improves the known quadratic lower bounds (e.g. \cite{KdW04, Wood07}). While  a modest improvement, we expect that the new techniques introduced in this work will be useful for further progress on lower bounds of locally correctable and decodable codes with more than 2 queries, possibly over other fields as well.

Our proof introduces several new ideas to existing lower bound techniques, several of which work over every field. At a high level, our proof has two parts, {\it clustering} and {\it random restriction}. 

The clustering step uses a powerful theorem of Barthe from convex geometry. It can be used (after preprocessing our LCC to be {\it balanced}), to apply a basis change (and rescaling) of the vectors, so that the resulting unit vectors become {\it nearly isotropic}. 
This together with the fact that any LCC must have many `correlated' pairs of points, lets us deduce that the vectors must have a surprisingly strong geometric clustering, and hence also combinatorial clustering with respect to the spanning triples. 

In the restriction step, we devise a new variant of the dimension reduction technique used in previous lower bounds, which is able to take advantage of the combinatorial clustering structure above. The analysis of our random projection method reduces to a simple (weakly) random graph process, and works over any field.

%One contribution is a new random restriction argument, whose dimension reduction effect leads to analyzing a simple (weakly) random graph process. Another contribution highlights the importance of understanding {\em balanced} codes, namely ones in which no `large' sub-code is contained in disproportionally `low' dimension. These were known to be the "hardest case" in that standard iterative methods can peel off subcodes of unbalanced codes get the lower bound. We prove further useful properties of balanced codes. In particular, they support a distribution on the bases of the generator matrix, in which most vectors have "high" probability to be in a basis. 

%At this point comes the field-specific component. Over the Reals, the last (combinatorial) property implies, by a powerful theorem of Barthe, a (geometric) property, namely that after appropriate basis change (and rescaling), the vectors become "nearly" isotropic: they correlate roughly equally with every direction in space. We analyze this geometric structure and show that, when $d << \sqrt{n}$, it implies  a (surprisingly strong) "clustering" structure of the collections of spanning triples. This structure in turn is used in analyzing the new random restriction argument mentioned above, to derive a contradiction.

\end{abstract}

\thispagestyle{empty}

\newpage
\pagenumbering{arabic}

\section{Introduction}

Locally-correctable codes (sometimes under different names of program self-correctors or random self-reductions), abbreviated LCCs, have the property that each symbol of a corrupted codeword can be recovered, with high probability, by randomly accessing only a few other symbols. LCCs have played a key role in important developments within several (impressively) diverse areas of theoretical computer science, which we briefly summarize below. 

Blum and Kannan~\cite{Blum-Kannan} introduced the idea of probabilistic, local correction for the purpose of program checking. With the follow-up papers~\cite{BLR93} on linearity testing and \cite{RS96} on low-degree testing this sequence inaugurated the field of Property Testing and Sublinear Algorithms. The realization of \cite{Lipton90, BeaverF90}, that  Reed-Muller codes (namely low-degree multivariate polynomials) are locally correctable, gave the first random self-reducibility examples of very hard functions like the Permament, and this average-case to worst-case complexity reduction was useful for pseudo-random generators~\cite{BFNW93}.
It further lead (with many more ideas) to the celebrated sequence of characterizations of the power of probabilistic proofs, $IP=PSPACE$ by \cite{LFKN,Shamir}, $MIP=PSPACE$ by \cite{BFL} and $PCP=NP$ by \cite{AS98,ALMSS}. Close cousins of LCCs, Locally-{\em Decodable} Codes (LDCs)\footnote{In LDCs one needs to locally recover {\em only} $d$ linearly independent coordinates (equivalently, the message) from the corrupted codeword, rather than all $n$ of them}, formally introduced in \cite{KT00} but having their origins in these earlier works, were key to Private Information Retrieval and other models of secure delegation of computation (see e.g. \cite{CKGS}). Dvir \cite{Dvir10} has shown the sufficiently strong lower bounds on LCCs  would yield explicit rigid matrices, which are related, via the work of \cite{Val77} to circuit complexity\footnote{While work of \cite{KSY11} shows that, over small finite fields, this approach could not give super linear circuit lower bounds, the approach might still be valid over large fields.}. While this has not materialized yet, it motivated the invention of {\em multiplicity codes} by \cite{KSY11} which are new LCCs of  high rate, and turn out to yield optimal list-decodable codes as well \cite{Kopparty12} . Finally, since the work of \cite{DvirShpilka06}, LDCs and LCCs have played a  role in understanding basic problems in Polynomial Identity Testing and established its connection to problems in Incidence Geometry, e.g \cite{KS09, BDWY11, DSW12}.  
  
The most important parameters of LCCs are the number of queries, $q$, made by the correcting algorithm, and the block length $n$ as a function of the message length (or dimension, for linear codes) $d$, where we fix corruptions to some small fixed fraction, say $1\%$. For upper bounds, the best constructions we have are still based on Reed-Muller codes\footnote{For the weaker LDCs there are far better constructions, based on the work of Yekhanin and Efremenko \cite{Yek08,Efr09,DGY11}, but these are not known to be locally correctable.} which exist only over   finite fields.  For $q$ queries these require block length about $\exp(d^{1/(q-1)})$. Indeed most applications require the block-length $n$ to be polynomial in $d$ and hence using these codes forces the number of queries to be at least logarithmic. Finding better codes, and in particular constant query, polynomial block-length LCCs, has been a major challenge, and this challenge naturally turns attention to the limits of constant query LCCs and LDCs. 

On the lower bound front, relatively little is known to rule out the feasibility of the challenge above. We shall restrict ourselves to {\em linear} codes\footnote{Some of the results below are known also for non-linear codes} over some field $\F$, namely when the set of codewords is a subspace of $\F^n$ of dimension $d$, and denote $q$-LCCs such locally-correctable codes with $q$ queries.  It is easy to see that 1-LCCs do not exist over any field. The first set of interesting results came for 2-LCCs, and here strong lower bounds are known through a variety of techniques. An exponential $n > 2^{\Omega (d)}$ lower bound via isoperimetric/entropy methods for 2-LCCs over $\F_2$ follows from the ones for the (weaker) LDCs \cite{GKST06, KdW04, DvirShpilka06} and is matched by the Hadamard code whose generating matrix is composed of all binary vectors over $\F_2$. Strangely, while these vectors provide an LDC over {\em every} field, they fail to be an LCC except in $\F_2$. This gap was first explained in \cite{BDWY11,DSW12} who showed that over the Real numbers (and indeed even  finite fields of very large characteristic), LCCs simply do not exist! For every error-rate $\delta$ the dimension $d$ for which such codes exist cannot exceed $\poly(1/\delta)$. The proofs in \cite{BDWY11,DSW12} use a combination of geometric, analytic and linear-algebraic techniques, and give quantitative form to  known qualitative point-line incidence theorems. Tighter bounds of $n> p^{\Omega (d)}$ over finite fields of prime size $p$ were proved in \cite{BDSS11} using methods from arithmetic combinatorics, matching the trivial construction of taking all vectors in $(\F_p)^d$.   
 
For $q\geq 3$ the known lower bounds are far weaker, and practically only one lower bound technique is known: random restrictions of the given code which reduce the number of queries from $q$ to 2 or 1, appealing to the lower bounds above. This technique was introduced for LDCs by Katz and Trevisan \cite{KT00}, and trivially holds for (the stronger) LCCs as well. The best bounds known are due to \cite{KdW04, Wood07}, which show that $q$-LDCs must satisfy $n = \tilde\Omega \left(d^{1 + 1/(\lceil q/2\rceil-1)}\right)$ for every  $q\geq 3$. So, in particular, the best lower bound for 3-LDCs (or LCCs) is the quadratic $n = \tilde\Omega(d^2)$ (for linear codes the $\tilde \Omega$ was replaced by $\Omega$ in \cite{Wood12}). This quadratic bound has established itself as somewhat of a `barrier' in that it can be obtained in several different ways and going beyond it seems to require new ideas.

 Our main result is breaking this quadratic barrier for 3-LCCs over the Real numbers. Namely, we prove that for some fixed constant\footnote{We did not make an attempt to optimize the constant $\barconst$, but the proof gives some $\barconst > .001$ } $\barconst >0$ every linear 3-LCC over the Reals must satisfy $n = \Omega(n^{2+\barconst})$, even when the error parameter $\delta$ is allowed to be polynomially small in $n$ . To this end, we introduce several new ideas and techniques, which we hope will lead to further progress. Some of our ideas are general enough to work over  any field, while others are specially tailored for the Reals. We briefly discuss now the main sources for our improvement over the known quadratic lower bound. A more detailed overview of the proof is given after the formal statement of the theorem in the next section.

\subsection*{Clustering and restrictions}

A linear 3-LCC over $\F$ may be viewed as a set  $V \subset \F^d$ of $n$ vectors (which form the generating matrix of the code), together with  $n$ collections $M_v$, one for each $v\in V$. Each $M_v$ is a matching of $\delta n$ disjoint triples from $V$, and each of the triples in $M_v$ spans $v$. This structure is easy to deduce for linear codes from the more traditional definition using a randomized decoder (cf. Definition~\ref{def-lccdecoder}).

We now informally describe a way to obtain a  quadratic lower bound on $n$, which uses random restrictions to reduce the dimension of the code. Pick a set $A$ of size about  $\sqrt{n}$ of vectors from $V$ at random. Then, take a linear projection whose kernel is exactly the span of the vectors in $A$ and apply it to the elements of $V$. Notice that in expectation, for every $v\in V$, a pair of points in $A$ will be contained in some triple in $M_v$. Thus, after the projection, the 3rd point in that triple will become the same as $v$ (up to scaling). As this happens to every point, we expect $V$ to shrink by a factor of 2! Repeating this process logarithmically many times will shrink $V$ completely, revealing that its original dimension could not have been larger than $\sqrt{n}\log n$, giving a near quadratic relation $n \geq d^2/\log d$. We note that the proofs appearing in the literature are somewhat different then the one we just  described. Indeed, there are several possible  ways of using a random restriction argument to get a quadratic bound (up to poly-logarithmic factors) for linear 3-LCCs. The argument above is new to this paper, and is indeed a simplified variant of our actual proof, which improves its analysis over the Reals.

It is not hard to see that if the collection of triples in all of matchings $M_v$ were  chosen at random, the analysis above could not be improved. But a random collection is far from being an LCC. Indeed, in contrast to standard codes, which exist in abundance and a random subspace is one with high probability, locally correctable (or decodable, or testable) codes are extremely rare and structured. This raises the question of what other structural properties are imposed on the matchings $M_v$ in an LCC. In this paper we reveal a new such property, {\em clustering}, at least when the  underlying field is the Reals\footnote{The actual proof requires several extra conditions on the code, which can be obtained via a sequence of reductions.}.  We conclude with a simplified  description of this clustering property, how it is obtained, and how it enables better analysis of the random restriction process.

A collection $M_v$ of matchings of triples is said to be {\em clustered} if there are about $\sqrt{n}$ subsets $S_1,\cdots S_{\sqrt{n}}$ of $V$, each of size about $\sqrt{n}$, such the {\em every} triple in {\em every} matching $M_v$ has a pair in one of these sets. Note that such a configuration is extremely far from random. Indeed, as these sets have at most $n^{3/2}$ pairs between them, many of the triples (of different matchings) share pairs (a typical pair exists in about $\sqrt n$ triples!). Note that this cluster structure is completely combinatorially described.

%Why should the triples in a 3-LCC admit such a clustering? The main observation is that, over the Reals, a small linearly dependent subset, such as a 4-tuples composed of $v$ and a triple from $M_v$, must contain a pair which is significantly correlated (say, with inner product at least $1/4$ for said example). Thus, a 3-LCC must contain many correlated pairs. A powerful result of Barthe from convex geometry allows us to deduce that, after a carefully chosen change of basis, this can only happen if the points in $V$ are {\em geometrically} clustered: They can be partitioned  into roughly $\sqrt(n)$ small balls of small radius. The correlations then must arise from  triples containing a pair in one of the (geometric) clusters. Thus this geoemtric clustering actually gives rise to a combinatorial clustering of the spanning triples. Though we show such a clustering only for LCCs over the Reals, such a result might be true over every field. Once we have the clustering, the rest of the argument is field independent. 

Why should the triples in an  3-LCC admit such a clustering? The main observation is that, over the Reals, a small linearly dependent subset, such as a 4-tuples composed of $v$ and a triple from $M_v$, must contain a pair which is significantly correlated (say, with inner product at least $1/4$ for said example). Thus, a 3-LCC must contain many correlated pairs. On the other hand, a powerful result of Barthe from convex geometry allows us to deduce that, after a carefully chosen change of basis, the vectors of our code are almost isotropic, namely point roughly equally in all directions in space. This implies that most pairs are hardly correlated at all. These two seemingly contradicting structures can exist only if the points in $V$ are {\em geometrically} clustered: delicate analysis shows that they can be partitioned  into roughly $\sqrt(n)$ small balls of small radius. The correlations then must arise from  triples containing a pair in one of the (geometric) clusters. 

Why does clustering help? Lets return to the random restriction and projection argument above, but lets pick now the set $A$ as follows. First pick one of the clusters $S_i$ uniformly at random, and inside it pick $A$ at random of size about $n^{1/4}$. The clustering ensures that this much smaller set has a pair intersecting each of the matchings $M_v$ in expectation (due to the fact that a typical pair in a typical cluster participates in $\sqrt n$ matchings). So a much smaller set $A$ suffices to create the same effect after projection, namely a shrinking of the set $V$ by a factor of 2. Again a logarithmic number of such restrictions is likely to shrink $V$ completely, giving a dimension upper bound of $n^{1/4}\log n$, and yielding the lower bound $n\geq d^4/\log d$.
We note again that this part works over any field, as long as the triples are clustered.

\paragraph{`Balanced' codes:} A recurring notion in our proof the that of an LCC in which no large subset of the coordinates lies in a subspace of significantly lower dimension. One can think of such codes as being `balanced' in the sense that they cannot be `compressed' (by projecting the large set of low dimension to zero). Our proof contains a sequence of reductions, used to obtain certain conditions that are used in the clustering and restriction steps. Each of these reductions can only be carried out if the code is `balanced' and this property is used in several different ways in the proof. If the code is not `balanced' we can use an iterative argument that projects the large low-dimensional subset to zero. We find this condition of being balanced a very natural one in the context of LCCs (and other codes) and hope it could be useful as a conceptual tool in future works.

\paragraph{Organization:} In Section~\ref{sec-defs} we state our results formally. Then, in Section~\ref{sec-overview} we provide a more detailed and technical overview of the proof. The proof of the main theorem is given in Sectins~\ref{sec-prelim} to \ref{sec-mainproof}. The organization of the different sections of the proof is given at the end of Section~\ref{sec-overview}.

\paragraph{Acknowledgments}
We are grateful to Boaz Barak, Moritz Hardt and Amir Shpilka for their contribution in early stages of this work. In particular, we thank Moritz Hardt for introducing us to Barthe's work.

\section{Definitions and results}\label{sec-defs}

For a string $y \in \F^n$, we define $w(y)$ to be the number of nonzero entries in $w$. A q-matching $M$ in $[n]$ is defined to be a set of disjoint unordered $q$-tuples (i.e. disjoint subsets of size $q$) of $[n]$ . 

\begin{define}[Linear $q$-LCC, decoder definition]\label{def-lccdecoder}
	A linear $(q,\delta)$-LCC of dimension $d$ over a field $\F$ is a $d$ dimensional linear subspace $U \subset \F^n$ such that there exists a randomized
	decoding procedure $D : \F^n \times [n] \mapsto \F$ with the following properties:
	\begin{enumerate}
	\item For all $x \in U$, for all $i \in [n]$ and for all $y \in \F^n$ with $w(y) \leq \delta n$ we have that $D\left( x + y, i\right) = x_i$ with probability at least $3/4$ (the probability is taken only over the internal randomness of $D$).
	\item For every $y \in \F^n$ and $i \in [n]$, the decoder $D(y,i)$ reads at most $q$ positions in $y$.
	\end{enumerate}
\end{define}

\begin{define}[Linear $q$-LCC, geometric definition]\label{def-linearlcc}
Let $V = (v_1,\ldots,v_n) \in (\F^d)^n$ be a list of $n$ vectors
spanning $\F^d$. We say that $V$ is a linear {\em
$(q,\delta)$-LCC} in geometric form if for every $v \in V$ there exists a $q$-matching $M_v$ in $[n]$ of size $\geq \delta n$ such that for every $q$-tuple $\{j_1,\ldots,j_q\} \in M_v$ it holds that $v \in \span\{v_{j_1},\ldots,v_{j_q}\}$.
\end{define}

It is well known that any linear $(q,\delta)$-LCC (over any field) can be converted into the geometric form given above by replacing $\delta$ with $\delta/q$. The transformation is simple: take $v_1,\ldots,v_n \in \F^d$ to be the rows of the generating matrix of $U$. Clearly, this does not change the dimension of the code.

In our results we will assume that the error parameter $\delta$ is not too large. Specifically, we will require that $n \geq (1/\delta)^{\omega(1)}$. This condition can be replaced with $n \geq (1/\delta)^{C}$ for a sufficiently large absolute constant $C$ which can be calculated from  the proof. 

We now state our main result which bounds the dimension of 3 query LCC's when the underlying field is $\R$. 

\begin{THM}[Main Theorem]\label{thm:lccbound}
There exists an absolute constant $\barconst > 0$ such that if $V = (v_1,\ldots,v_n) \in (\R^d)^n$ is a linear $(3,\delta)$-LCC and $n \geq (1/\delta)^{\omega(1)}$, then $$
d = \dim(V) \leq n^{1/2- \barconst}
$$ 
%with $\barconst$ an absolute constant larger than zero.
\end{THM}

\section{Proof overview: `Cluster and Restrict' paradigm}\label{sec-overview}

From a high level, our proof is divided into two conceptually distinct steps:
\begin{enumerate}
	\item {\em Clustering step:} Show that the triples used in the  matchings $M_v, v\in V$ are `clustered' in some precise sense (described below). 
	\item {\em Restriction step:} Use the clustering to find a large subset of $V$ that has low dimension. The name of this step is due to the fact that it uses a random restriction argument (projecting a random subset to zero).
\end{enumerate}

Combining these two step (in Lemma~\ref{lem-subsetlow}) we get that $V$ must have a large subset  (of size roughly $\Omega(n)$)  with low dimension (at most $n^{1/2 - \barconst}$). Using this to prove a  global dimension bound on $V$ (as in Theorem~\ref{thm:lccbound}) is done using a standard amplification lemma (Lemma~\ref{lem:iterate}) similar to that in ~\cite{BDWY11,BDSS11}. For simplicity, we will use big `O' notation to hide constants depending on $\delta$ (only for this overview).

We now describe each of these steps in more detail. The fact that $V$ is a code over $\R$ is only used in the clustering step.  The restriction step works over any field, provided that the triples are already clustered.  A recurring theme in the proof is that we are always free to assume that $V$ does not have a large subset of low dimension. Another recurring operation is `mapping a subset $U$ of $V $ to zero'. By this statement we mean: pick a linear map $A$ whose kernel is $\span(U)$ and apply it to all the elements of $V$. We will use the simple fact that, if $\dim(U) = r$ and $\dim(V) = d$ then $\dim(A(V))$ is at least $d - r$, where $A(V)$ is the list of vectors $A(v), v \in V$.

\subsection{Clustering Step:}

The clustering step is given by Lemma~\ref{lem-reducecluster} which we state now in an informal form. We will elaborate below on the two conditions necessary in the lemma. Recall that $V$ is associated with $n$ $3$-matchings $M_v, v\in V$ used in the decoding.

{\noindent  \bf Lemma~\ref{lem-reducecluster}. [Informal] }{\it Suppose $V$ is a $(3,\delta)$-LCC that satisfies the `well-spread' condition and the `low triple multiplicity' condition and suppose that $d > n^{1/2 - \barconst}$. Then there are subsets $S_1,\ldots,S_m \subset V$ (not necessarily disjoint) so that
\begin{enumerate}
	\item For each $i \in [m]$, $|S_i| \leq O(n^{1/2+\barconst})$.
	\item $\Omega(n^{1/2-\barconst}) \leq m \leq O(n^{1/2 + \barconst})$.
	\item Each triple in each matching $M_v$ has two of its elements in one of the sets $S_i$.
\end{enumerate}
}

Before we explain the two conditions in the lemma of being `well-spread' and having `low triple multiplicity', notice that the existence of sets $S_1,\ldots,S_m$ as above is something that does not hold for a `typical' family of $\Omega(n^2)$ triples. In fact, if the triples were chosen at random there would not be such sets with probability close to one. Referring to the sets $S_i$ as `clusters' is also justified by the fact that they actually form clusters in $\R^d$ (i.e., they are all correlated with some fixed point). This geometric fact, however, is not used anywhere in the proof-- all we need is the combinatorial structure. We now explain the two conditions on the code $V$ mentioned in the lemma:
\begin{itemize}
	\item {\bf Well-spread condition:} The vectors $v_1,\ldots,v_n$ comprising $V$ should be `well-spread'. Observe that WLOG by a suitable scaling to each vector, we can assume that the vectors $v_1,\ldots,v_n$ are unit vectors, and we will make this assumption. Formally, we require that for every unit vector $w \in \R^d$ we have $\sum_{i \in [n]} \ip{v_i}{w}^2 \leq  O(n^{1/2+\barconst})$. This means, in particular, that every small ball can contain at most $O(n^{1/2+\barconst})$ vectors.  Clearly, a general LCC $V$ does not have to satisfy this condition. For example, if $V$ has a large subset that lies in low dimension, such a statement cannot hold (using pigeon hole argument on the unit circle in low dimension).  We  are able, however, to reduce to this case using Lemma~\ref{lem-reducewell}, which uses a powerful result of Barthe (Lemma~\ref{cor-barthe2}) that is developed in Section~\ref{sec-barthe}. Roughly speaking, Barthe's theorem can be used to show that, unless $V$ has a large subset in low dimension, there is an invertible linear map $M$ on $\R^d$ so that, if we replace each $v_i$ with $Mv_i/\|Mv_i\|$, the well-spread condition is satisfied. The proof of this result (part of which appear in Section~\ref{sec-barthe}) uses tools from convex geometry. We derive a particularly convenient form of Barthe's theorem as Theorem~\ref{thm-bartheconvenient} which might be of independent interest.
	\item 
	{\bf Low triple-multiplicity condition:} This condition  requires that a single triple does not appear in `too many' (roughly $n^{O(\barconst)}$) different matchings. In Section~\ref{sec-reducelowmult} we prove Lemma~\ref{lem-reducelowmult} which shows how to reduce to this case, assuming $V$ does not have a large low dimensional subset. The reduction uses the fact that if a single triple is used in too many matchings, then projecting the elements in this triple to zero causes many other points to go to zero. If a point $v$ is mapped to zero as a result, and if $v$ is used in many triples (say $\Omega(n)$) all of these triples `become' pairs when $v$ maps to zero. Using this observation, we show that we can send a relatively small number of points to zero and construct a $2$-query locally decodable code (LDC) of relatively high dimension. We then apply the known bounds for $2$-query LDCs (these are variants of LCCs and described in Section~\ref{sec-prelim}) to get a contradiction. This reduction is also field independent and does not use any properties of the real numbers.

\end{itemize}

The main observation leading to clustering is that we can assume, w.l.o.g that all triples $(i,j,k) \in M_v$ are so that the three vectors $v_i,v_j,v_k$ are almost orthogonal to $v$. This follows directly from the `well spread' condition by upper bounding the number of vectors correlating with $v$ and discarding the corresponding triples from $M_v$ (for each $v \in V$). Once we have this condition, we observe that since $v,v_i,v_j,v_k$ are linearly dependent and, since $v$ is not correlated with the other three vectors, we must have that $v_i,v_j,v_k$ are close to being in a two dimensional plane (recall that these are all unit vectors). This means that in each triple there must be two elements that are correlated with each other! This is already a non trivial fact, in particular since we know (by the well spread condition) that each point cannot be correlated with many other points. 

Proceeding with a more careful analysis of the different types of triples that can arise, and using some graph theoretic arguments, we arrive at the required clusters. In this step we use the bound on the maximum triple multiplicity. 
 
 Note that the clustering lemma implies that there are many pairs in $V \times V$ that appear in many triples. This is due to the simple upper bound of $n^{1.5+O(\barconst)}$ on the total number of possible pairs in all of the clusters $S_1,\ldots,S_m$ and the fact that together they cover pairs from a quadratic number of triples. This should be contrasted with the results of \cite{BDWY11,DSW12} which prove strong lower bounds for $q$-LCC's (for any constant $q$) in which every pair is in a bounded number of triples (these are called `design' LCCs).

\subsection{Restriction Step:}

The restriction step (given in Lemma~\ref{lem-clusterlowdim}) shows that if $V$ satisfies the clustering condition (given in Lemma~\ref{lem-reducecluster}) then it contains a large subset in low dimension. We now state a simplified form of this lemma.

{\noindent  \bf Lemma~\ref{lem-clusterlowdim}. [Informal] }{\it Let $\F$ be a field. Let $V = (v_1,\ldots,v_n ) \in (\F^d)^n$ be a $(3,\delta)$-LCC with  matchings $M_v, v \in V$. Suppose there exists sets $S_1,\ldots,S_m \subset [n]$ as in Lemma~\ref{lem-reducecluster}, clustering the triples in the matchings $M_v$. Then, there is a subset $V' \subset V$ of size $|V'| \geq (\delta/2)n$ and dimension at most $n^{1/2 - \barconst}$.
}

This step is called the `restriction step' since it uses the `clusters' $S_1,\ldots,S_m$ found in the clustering step to show (Lemma~\ref{lem-randomrestrict})  that there is a  small set $U \subset V$ (of size roughly $n^{1/4+7\barconst}$) such that, projecting all elements of $U$ to zero, reduces the dimension of $V$ to at most $n^{10\barconst}$. This will imply a dimension bound of $n^{1/4+7\barconst} + n^{10\barconst}$ on the initial dimension of $V$ (the reason we do not get a $n^{1/4+7\barconst}$ upper bound on the dimension of $V$ is due to the clustering step). 

The starting point for the proof of this lemma is the following simple observation: If $v$ is spanned by a triple $(v_i,v_j,v_k)$, then projecting two elements of that triple, say $v_i,v_j$, to zero makes the two vectors $v,v_k$ proportional to each other (this uses the fact that $v$ is not spanned by any proper subset of the triple, and we can easily reduce to this case). Now, suppose that there are $t$ triples in the code that have at least two element in $U$. Then projecting $U$ to zero makes makes $t$ pairs of vectors proportional to each other (as in the $v,v_k$ example). Consider the graph on vertex set $V$ in which we add an edge for each proportional pair $v,v_k$ obtained by sending a pair $v_i,v_j \in U$ in a triple $(v_i,v_j,v_k) \in M_v$ to zero. Since the property of being proportional to each other is an equivalence relation on $\R^d$, we can bound the dimension of $V$ after projecting $U$ to zero by the number of connected components of the graph.

This leaves us with the task of finding a set $U$ so that the resulting graph has at most $n^{10\barconst}$ components. To find such a $U$ we use a probabilistic argument. We will pick $U$ at random according to a particular distribution and then argue that the expected number of connected components is small. To pick the random $U$ we proceed in $r \sim n^{4\barconst}$ steps as follows: In each step pick one of the clusters $S_i$ at random and then pick a random subset of $S_i$ of size $\sim n^{1/4+3\barconst}$ at random. The union of these sets will be $U$. The upper bound on the expected number of components is derived by considering the (expected) reduction in the number of connected components in each of the $r$ steps. Consider some connected component and let $v$ be some vector in it. We can assume the component is not too large, since the number of large components is trivially bounded (large being close to $n^{1-\barconst}$). Since each $M_v$ is a matching, the random choice of the vectors in the $i$'th step will (with good probability) add an edge to $v$ with a neighbor that is not likely to land in the connected component containing $v$. Hence, with good probability the connected component will `merge' with another component. Carefully analyzing this process gives us the required bound.

\subsection{Proof Organization}

We begin with some general preliminaries and notations in Section~\ref{sec-prelim}. In Section~\ref{sec-barthe} we describe (and sketch the proof of) Barthe's theorem which is used in Section~\ref{sec-reducewell} to reduce to the case that the points in $V$ are well-spread. In Section~\ref{sec-reducelowmult} we show how to reduce to the case that $V$ has low triple multiplicities. Section~\ref{sec-reducecluster} contains the proof of the clustering step and Section~\ref{sec-clusterlowdim} contains the proof of the restriction step. Finally, in Section~\ref{sec-mainproof} we show how to put all the ingredients together and prove Theorem~\ref{thm:lccbound}.

%%%%%%%%%%%%%%%%%%%%%%%%%%%%%%%%%%%%%%%%%%%%%%%%%%%%%%%%%%%%%%%%%%%%%%%%%%%%%
%%%%%%%%%%%%%%%%%%%%%%%%%%%%%%%%%%%%%%%%%%%%%%%%%%%%%%%%%%%%%%%%%%%%%%%%%%%%%%
%%%%%%%%%%%%%%%%%%%%%%%%%%%%%%%%%%%%%%%%%%%%%%%%%%%%%%%%%%%%%%%%%%%%%%%%%%%%%%
%%%%%%%%%%%%%%%%%%%%%%%%%%%%%%%%%%%%%%%%%%%%%%%%%%%%%%%%%%%%%%%%%%%%%%%%%%%%%%
%%%%%%%%%%%%%%%%%%%%%%%%%%%%%%%%%%%%%%%%%%%%%%%%%%%%%%%%%%%%%%%%%%%%%%%%%%%%%%

\section{General Preliminaries}\label{sec-prelim}

\subsection{Choice of notation}
\paragraph{Lists vs. multisets: }The reason we are treating $V$ as a list and not as a set is that $V$ might have repetitions. For instance $u$ and $v$ might be distinct elements in the list $V$, but might correspond to the same vector in $\F^d$. The repetition corresponds to the fact that there might be repeated columns in the generator matix of the code, which may potentially make the property of local correction easier to satisfy. Indeed in the recent lower bounds for 2-query LCCs~\cite{BDSS11, BDWY11}, handling the fact that there might be  repetitions added significant complexity to the proofs of the lower bounds. In the current paper too we deal with repetitions by treating $V$ as a list. An equivalent treatment would be to treat $V$ as a multiset, and we make no distinction between these notions. We think of a multiset as an ordered list of elements which might contain repeated elements.  If $A$ is a multiset/list, we call $B$ a subset of $A$ if $B$ is another multiset/list obtained by taking a subset of $A$. We will say that $B$ and $C$ are {\em disjoint} subsets of $A$ if they are both obtained from sub-lists on disjoint subsets of the indices. When referring to the {\em size} of a multiset we will always count the number of elements {\em with} multiplicities (unless we state explicitly that we are counting {\em distinct} elements). 

Although we defined a matching to be a set of tuples in $[n]$, when we are dealing with a specific list $V = (v_1,\ldots,v_n)$, we might identify a tuple $(j_1,\ldots,j_q)$ of a matching with the tuple $(v_{j_1},\ldots,v_{j_q})$, and we use these two notions interchangably. 
Moreover, a matching $M_v$ denotes the matching corresponding to a particular element $v \in V$, and if $u$ and $v$ are different elements of $V$, even if they correspond to the same vector in $\F^d$, then $M_u$ and $M_v$ could be different matchings.

\subsection{Basic operations on LCCs}

For a list $V \in (\R^d)^n$ we denote by $\span(V)$ the subspace spanned by elements of $V$ and by $\dim(V)$ the dimension of this span. 

The following simple claim shows that a sufficiently large subset of an LCC is also an LCC.
\begin{claim}\label{cla-refinelarge}
	If $V = (v_1, \dots, v_n) \in (\F^d)^n$ is a $(3,\delta)$-LCC and $U \subset V$ is of size $|U| \geq (1 - \delta/2)n$ then $U$ is a $(3,\delta/2)$-LCC of the same dimension as $V$. Moreover, if $M_v, v \in V$ are any matchings used in the decoding of $V$ then we can take the matchings for the new code $U$ to be subsets of the old matchings.
\end{claim}
\begin{proof}
Observe that in each matching $M_v$, there are at most $(\delta/2)n$ triples that contain an element outside $U$. Thus, in $U$ we could construct matchings of size $(\delta/2)n \geq (\delta/2)|U|$. The claim about the dimension follows from the fact that $U$ contains triples spanning all of the elements of $V$ (not just those in $U$).
\end{proof}

Another simple observation is that applying an invertible linear map to the elements of $V$ preserves the property of being an LCC. 

\begin{obs}\label{obs-linearmap}
	If $V = (v_1,\ldots,v_n) \in(\R^d)^n$ is a $(3,\delta)$-LCC then, for any invertible linear map $M: \R^d \mapsto \R^d$ the list $\hat V = (\hat v_1,\ldots,\hat v_n) \in (\R^d)^n$, with $\hat v_j = \frac{Mv_j}{\|Mv_j\|}$, is also a $(3,\delta)$-LCC.
\end{obs}

\subsection{Lower bounds for $2$-query LDCs}\label{sec-ldc}

One of the ingredients in the proof will be a strong (exponential) lower bound on the length of linear 2-query Locally Decodable Codes (LDCs), which are weaker versions of LCCs. As with LCCs there are two ways of defining LDCs.

\begin{define}[linear $q$-LDC, decoder definition]
A linear $(q,\delta)$-LDC over a field $\F$ is a linear $d$-dimensional subspace $U \subset \F^n$, and a set of $d$ coordinates $j_1, j_2, \ldots j_n \in [n]$ such that the projection of $U$ on to those $d$ coordinates is full dimensional\footnote{If the LDC was systematic, then the first $d$ coordinates would suffice.}, and such that there exists a randomized
	decoding procedure $D : \F^n \times [d] \mapsto \F$ with the following properties:
	\begin{enumerate}
	\item For all $x \in U$, for all $i \in [d]$ and for all $y \in \F^n$ with $w(y) \leq \delta n$ we have that $D\left( x + y, i\right) = x_{j_i}$ with probability at least $3/4$ (the probability is taken only over the internal randomness of $D$).
	\item For every $y \in \F^n$ and $i \in [d]$, the decoder $D(y,i)$ reads at most $q$ positions in $y$.
	\end{enumerate}

\end{define}

Let $\{e_1, e_2, \ldots, e_d\}$ be the set of standard basis vectors in $\R^d$. 

As with LCCs, taking the rows of the generating matrix (and possibly applying an invertible linear map that sends them to the $e_i$s) allows us to move to the geometric form. This might require us to replace $\delta$ with $\delta/q$.

\begin{define}[linear $q$-LDC, geometric definition]
Let $V = (v_1,\ldots,v_n) \in (\F^d)^n$ be a list of $n$ vectors
spanning $\F^d$. We say that $V$ is a linear {\em
$(q,\delta)$-LDC} in geometric form if for every $i \in [d]$ there exists a q-matching $M_i$ in $[n]$ of size $\geq \delta n$ such that for every $q$-tuple $\{v_{j_1},v_{j_2},\ldots, v_{j_q}\} \in M_i$ it holds that $e_i \in \span\{v_{j_1},v_{j_2},\ldots, v_{j_q}\}$. We denote by $d = \dim(V)$.
\end{define}

\begin{thm}[lower bounds for 2-LDC \cite{DvirShpilka06}]\label{thm-2LDC}
Let $\delta \in [0,1]$, $\F$ be a field, and let $V = (v_1,\ldots,v_n) \in  (\F^d)^n$ be a linear $(2,\delta)$-LDC in geometric form. Then $$n \geq 2^{\frac{\delta d}{16} -1}.$$
\end{thm}

\subsection{Codes in regular form}
In the restriction step, it is convenient for us to assume that for each triple $(v_i, v_j, v_k) \in M_v$ each element of the triple is ``used" in decoding to $v$. Indeed in Claim~\ref{cla-regular}, we show how we can easily reduce to this case provided that no large subset of $V$ is contained in a low dimnesional space. More precisely, for $x,y,z \in \R^d$, let us denote by $\span^*\{x,y,z\}$ the set of all elements of the form $\alpha x +  \beta y + \gamma z $ with $\alpha,\beta,\gamma \in \R$, such that $\alpha,\beta,\gamma$ are all nonzero.  
\begin{define}
	Let $V = (v_1,\ldots,v_n) \in (\F^d)^n$ be a $(3,\delta)$-LCC with decoding matchings $M_v, v\in V$. We say that $V$ (with these matchings) is in {\em regular} form if, in each triple $(x,y,z) \in M_v$ we have that $v \in \span^*\{x,y,z\}$.
\end{define}
,
\begin{claim}\label{cla-regular}
Let $V = (v_1,\ldots,v_n) \in (\F^d)^n$ be a $(3,\delta)$-LCC so that every subset $U \subset V$ of size $|U| \geq (\delta/2)n$ has dimension at least $\omega((1/\delta) \log(n))$. Then, there exists a $(3,\delta/4)$-LCC $\,\,V' \subset V$ of size $n' \geq (1 - \delta/2)n$, and dimension $d'=d$, that is in regular form. Moreover, given any matchings $M_v$ for the code $V$ we can take the new (regular) matchings $M'_v$ for $V'$ to be sub-matchings of the original ones.
\end{claim}
\begin{proof}
Call a triple $(x,y,z) \in M_v$ {\em bad} if there is a proper subset of it that spans $v$, i.e. $v \not \in \span^*\{x,y,z\}$. If there were $(\delta/2)n$ points $v \in V$, each with at least $(\delta/10)n$ bad triples in $M_v$, then we could use these bad triples to construct a $(2,\delta/10)$-LDC of size $\leq n$ decoding $\omega((1/\delta) \log(n))$ linearly independant elements of $V$. This would give a contradiction using Theorem~\ref{thm-2LDC} and the assumption on the dimension of any set of size $(\delta/2)n$ in $V$. Therefore, there are at most $(\delta/2)n$ points $v \in V$ with many ($\geq (\delta/10)n$) bad triples. Throwing away this set, and removing all triples containing them (as well as all bad triples from the other matchings) gives us the code $V'$ a required (as in Claim~\ref{cla-refinelarge}).
\end{proof}

%%%%%%%%%%%%%%%%%%%%%%%%%%%%%%%%%%%%%%%%%%%%%%%%%%%%%%%%%%%%%%%%%%%%%%%%%%%%%
%%%%%%%%%%%%%%%%%%%%%%%%%%%%%%%%%%%%%%%%%%%%%%%%%%%%%%%%%%%%%%%%%%%%%%%%%%%%%%
%%%%%%%%%%%%%%%%%%%%%%%%%%%%%%%%%%%%%%%%%%%%%%%%%%%%%%%%%%%%%%%%%%%%%%%%%%%%%%
%%%%%%%%%%%%%%%%%%%%%%%%%%%%%%%%%%%%%%%%%%%%%%%%%%%%%%%%%%%%%%%%%%%%%%%%%%%%%%
%%%%%%%%%%%%%%%%%%%%%%%%%%%%%%%%%%%%%%%%%%%%%%%%%%%%%%%%%%%%%%%%%%%%%%%%%%%%%%

\section{Barthe's theorem}\label{sec-barthe}

The main purpose of this section is to derive Lemma~\ref{cor-barthe2},  a result of F. Barthe \cite{Bar98} which, given a set of points sufficiently close to being in general position, finds a linear transformation that `moves' these points so that their `directions'  point in a close to uniform way. More precisely, for a set $U = (u_1, \ldots,u_n) \in (\R^d)^n$ let $\cB(U)$ be the set of all subsets of $[n]$ of size $d$ such that the corresponding vectors of $U$ form a basis of $\R^d$. Suppose that there is a distribution $\mu$ supported on $\cB(U)$ such when sampling a random basis from $\mu$, each element of $U$ is chosen with some good probability. Then there is an invertible linear transformation such that after normalizing, the new points are ``approximately isotropic".  This result is formalized in Lemma~\ref{cor-barthe2} which we state below:

%%%%%%%%%%%%
\begin{lem}\label{cor-barthe2}
	Let $U = (u_1, \ldots,u_n) \in (\R^d)^n$. Let $S \subseteq [n]$, and suppose $\mu$ is a distribution supported on $\cB(U)$ such that for all $j \in S$	
	$$\alpha \leq \Pr_{I \sim \mu}[ j \in I]$$ Then, there exists an invertible  linear map $M: \R^d \mapsto \R^d$ so that, denoting $\hat u_j = \frac{Mu_j}{\|Mu_j\|}$, we have for all unit vectors $w \in \R^d$
	$$
	\sum_{j\in S}  \ip{\hat u_j}{w}^2  \leq \frac{2}{\alpha}
	$$	
\end{lem}

%%%%%%%%%%%%

Observe that if the vectors are in general position then the uniform distribution on distinct $d$-tuples gives $\alpha = d/n$, in which case we would get $$
	\sum_{j\in [n]}  \ip{\hat u_j}{w}^2  \leq \frac{2n}{d}.
	$$	

One can just assume the lemma above which follows in a straightforward way from  from  \cite{Bar98}, and skip to the next section. However for completeness, we present a proof here. Before we give the proof, we first set up some notation.

For a finite set $S$, a distribution supported on $S$ is a function $\mu :S \mapsto [0,1]$ so that $\sum_{x \in S}\mu(x) = 1$. For two vectors $u,v \in \R^d$ we denote by $u \otimes v$ the tensor product of $u$ and $v$, namely the $d \times d$ matrix with entries $A_{ij} = u_iv_j$. We denote by $I_{d \times d}$ the $d \times d$ identity matrix. For $u \in \R^d$ we denote by $\|u\|$ the Euclidean (or $\ell_2$) norm.

\begin{define}[$\cB(U), \,\cK(U)$]
Let $U = (u_1, \ldots,u_n) \in (\R^d)^n$ be a list of $n$ points. Let $I \subseteq [n]$. We denote by $U_I = (u_i)_{i \in I}$ the sub-list of $U$ with indices in $I$. We denote by $$\cB(U) = \{ I \subset [n]\,\,|\,\, U_I \text{ is a basis of } \R^d \}$$ the set of index sets corresponding to sub-lists of $U$ of length $d$ which are linearly independent (and so span $\R^d$). For each $I \subset [n]$ we let $1_I \in \R^n$ denote the indicator vector of the set $I$. Finally we denote by $\cK(U) \subset \R^n$ the convex hull of the vectors $1_I$ for all $I \in \cB(U)$. We denote by $\cK(U)^o$ the relative interior of $\cK(U)$\footnote{The relative interior of a set is a subset of the points of the set that are not on the boundary of the set, relative to the smallest subspace containing the set}.
\end{define}

\begin{claim}[Properties of $\cK(U)$]\label{cla-propKU}
Let $U = (u_1, \ldots,u_n) \in (\R^d)^n$ be a list of $n$ points spanning $\R^d$. Let $\mu$ be a distribution supported on $\cB(U)$. For each $j \in [n]$, let $\gamma_j \in [0,1]$ be the probability that $j \in I$, when $I \subset [n]$ is sampled according to $\mu$. Then $\gamma = (\gamma_1,\ldots,\gamma_n)$ is in $\cK(U)$.
\end{claim}
\begin{proof}
The vector $\gamma$ is easily seen to be equal to the convex combination $$ \sum_{I \in \cB(U)}\mu(I)\cdot 1_I. $$
\end{proof}

\begin{thm}[\cite{Bar98}]\label{thm-barthe}
Let $U = (u_1, \ldots,u_n) \in (\R^d)^n$ be a list of $n$ points spanning $\R^d$ and let $\gamma = (\gamma_1,\ldots,\gamma_n) \in \cK(U)^o$. Then there exists a real invertible $d \times d$ matrix $M$ such that,  denoting $\hat u_j = \frac{Mu_j}{\|Mu_j\|}$, we have
\begin{equation}\label{eq-barthe}
\sum_{j=1}^n \gamma_j \cdot (\hat u_j \otimes \hat u_j) = I_{d \times d}
\end{equation}
\end{thm}
\begin{proof}
We will show how the proof follows from one of the propositions proved in \cite{Bar98} (whose proof we will not repeat here). The idea is to define a certain optimization problem parametrized by $\gamma$ and to show that the maximum is achieved for all $\gamma \in \cK(U)$. Then, the matrix $M$ will arise from equating the gradient to zero at the maximum and solving the resulting equations.

We start be defining the optimization problem. For $t \in \R^n$ we define $$ X = X(t) = \sum_{j =1}^n e^{t_j} \cdot (u_j \otimes u_j).$$ Notice that $X(t)$ has a positive determinant for all $t \in \R^n$, since $U$ spans $\R^d$. Let $f: \R^n \times \R^n \mapsto \R$ be defined as $$f(\gamma,t) = \ip{\gamma}{t} - \ln\det(X(t)). $$ The optimization problem is defined as $$ \phi^*(\gamma) = \sup_{t \in \R^n} f(\gamma,t). $$ We now state a claim from \cite{Bar98} which give sufficient conditions for the supremum $\phi^*(\gamma)$ to be realized.

\begin{claim}[Rephrased from Proposition 6 in  \cite{Bar98}]\label{cla-prop6}
If $\gamma \in \cK(U)^o$ then the supremum $\phi^*(\gamma)$ is achieved. That is, there exists $t^* \in \R^n$ such that $f(\gamma, t^*) = \phi^*(\gamma)$.
\end{claim}

Let $t^* \in \R^n$ be a maximizer given by the claim. 
We can now use the fact that the partial derivatives $\frac{\partial f(\gamma,t)}{\partial t_j}$ all vanish at the point $t^*$. Recall that $\frac{d}{ds}\ln\det(A) = \tr\left( A^{-1} \frac{d}{ds} A\right)$ at all points where $A$ is invertible \cite[Ch. 9, Thm. 4]{LaxBook}. Taking the derivative of $f$ at $t^*$ then gives:
$$
0 = \frac{\partial f(\gamma,t)}{\partial t_j}(t^*) = \gamma_j - \tr\left(X(t^*)^{-1} e^{t^*_j} (u_j \otimes u_j)\right).
$$
Since $X(t^*)^{-1}$ is positive definite, there exists a symmetric matrix $M$ so that $M^2 = X(t^*)^{-1}$. Plugging this into the last equation and using properties of the trace function, we get:
$$ 0 = \gamma_j - e^{t^*_j}\| Mu_j\|^2.  $$ This means that
$$ M^{-2} = X(t^*) = \sum_{j=1}^n \frac{\gamma_j}{\| Mu_j\|^2} \cdot (u_j \otimes u_j) = \sum_{j=1}^n \gamma_j \cdot \left( \frac{u_j}{\| Mu_j\|} \otimes \frac{u_j}{\| Mu_j\|}\right). $$ Multiplying by $M$ from both sides we get
$$ I_{d \times d} = \sum_{j=1}^n \gamma_j \cdot \left( \frac{Mu_j}{\| Mu_j\|} \otimes \frac{Mu_j}{\| Mu_j\|}\right)$$ as was required.
\end{proof}

% \begin{cor}\label{cor-barthe}
% Let $U = (u_1, \ldots,u_n) \in \R^d$. Suppose $\mu$ is a distribution supported on $\cB(U)$ such that for all $j \in [n]$	
% $$\alpha \leq \Pr_{I \sim \mu}[ j \in I] \leq \beta.$$ Then, there exists an invertible $d \times d$ matrix $M$ so that, denoting $\hat u_j = \frac{Mu_j}{\|Mu_j\|}$, we have for all unit vectors $v \in \R^d$
% \begin{equation}\label{eq-barthe-cor}
% \frac{1}{2\beta} \leq  \sum_{j=1}^n  \ip{\hat u_j}{v}^2  \leq \frac{2}{\alpha}
% \end{equation}
% \end{cor}
% \begin{proof}
% 
% \end{proof}

\begin{proof}[Proof of Lemma~\ref{cor-barthe2}]
	Let $\gamma \in \R^n$ be such that $\gamma_j =  \Pr_{I \sim \mu}[ j \in I] $ for all $j \in [n]$.  By Claim~\ref{cla-propKU}, $\gamma \in \cK(U)$. This means we can find $\gamma' \in \cK(U)^o$ of distance at most $\barconst$ from $\gamma$ for all $\barconst > 0$. Hence, we can choose $\barconst$ sufficiently small so that $\alpha/2 \leq \gamma_j'$ for all $j \in S$. Using Theorem~\ref{thm-barthe} we get that there exists an invertible $M$ so that
	$$ I_{d \times d} = \sum_{j=1}^n \gamma'_j (\hat u_j \otimes \hat u_j).$$ Multiplying by the column vector $w$ from the left and by the row vector $w^t$ from the right we get that\
\begin{equation*}
	1 = \ip{w}{w} = \sum_{j=1}^n \gamma'_j \ip{\hat u_j}{w}^2 \geq (\alpha/2)\sum_{j\in S} \ip{\hat u_j}{w}^2.
\end{equation*}	
	This completes the proof.	
\end{proof}

%%%%%%%%%%%%%%%%%%%%%%%%%%%%%%%%%%%%%%%%%%%%%%%%%%%%%%%%%%%%%%%%%%%%%%%%%%%%%
%%%%%%%%%%%%%%%%%%%%%%%%%%%%%%%%%%%%%%%%%%%%%%%%%%%%%%%%%%%%%%%%%%%%%%%%%%%%%%
%%%%%%%%%%%%%%%%%%%%%%%%%%%%%%%%%%%%%%%%%%%%%%%%%%%%%%%%%%%%%%%%%%%%%%%%%%%%%%
%%%%%%%%%%%%%%%%%%%%%%%%%%%%%%%%%%%%%%%%%%%%%%%%%%%%%%%%%%%%%%%%%%%%%%%%%%%%%%
%%%%%%%%%%%%%%%%%%%%%%%%%%%%%%%%%%%%%%%%%%%%%%%%%%%%%%%%%%%%%%%%%%%%%%%%%%%%%%

\section{Reducing to the well-spread case}\label{sec-reducewell}

In this section we prove a lemma saying that, when analyzing an LCC $V = (v_1,\ldots,v_n)$ over $\R$, we can assume  that the elements of $V$  are unit vectors pointing in `well spread' directions. The precise form of `well spread' is that given by Barthe's theorem (Lemma~\ref{cor-barthe2}).  More formally, the lemma will say that {\em any} list of vectors can be transformed into  `well-spread' list as long as it does not contain a large low dimensional subset. We formalize this result in Theorem~\ref{thm-bartheconvenient}. Below we state a lemma which basically follows as a corollary of the above theorem when the original list of vectors is an LCC. We first state and prove this lemma. 

\begin{lem}\label{lem-reducewell}
Let $V = (v_1,\ldots,v_n) \in (\R^d)^n$ be a $(3,\delta)$-LCC be so that any subset $V' \subset V$ with $|V'| \geq (\delta/4)n$ satisfies  $\dim(V') > 4\beta d$. Then, there exists a subset $U = (u_1,\ldots,u_{n'}) \subset V$ that is a $(3,\delta/2)$-LCC with $|U| = n' \geq (1 - \delta/2)n$, and an invertible linear map $M:\R^d \mapsto \R^d$ so that, denoting $\hat u_j = \frac{Mu_j}{\|Mu_j\|}$, we have for all unit vectors $w \in \R^d$. 
	\begin{equation*}
	\sum_{j\in [n']}  \ip{\hat u_j}{w}^2  \leq \frac{n}{\beta d}.
	\end{equation*}
\end{lem}

Recall that (Observation~\ref{obs-linearmap}) applying an invertible linear map to the elements of an LCC $V$ preserves the property of being an LCC. Hence, if we are aiming to prove that a $(3,\delta)$-LCC $V$ has a large low dimensional subset, we could use Lemma~\ref{lem-reducewell} to reduce to the case that the points of  $V$ are `well-spread'.

 %We begin with a trivial (but useful) observation:
%\begin{obs}\label{obs-linearmap}
%	If $V = (v_1,\ldots,v_n) \in(\R^d)^n$ is a $(3,\delta)$-LCC then, for any invertible linear map $M: \R^d \mapsto \R^d$ the list $\hat V = (\hat v_1,\ldots,\hat v_n) \in (\R^d)^n$, with $\hat v_j = \frac{Mv_j}{\|Mv_j\|}$, is also a $(3,\delta)$-LCC.
%\end{obs}

We will prove Lemma~\ref{lem-reducewell} using Lemma~\ref{cor-barthe2}. Recall that, Lemma~\ref{cor-barthe2} provides us with sufficient conditions under which a linear map $M$ as in the lemma exists. Namely, that there exists a distribution $\mu$ on spanning $d$-tuples of $V$ which hits each element in $V$ with probability not too small. We will show that, if this condition does not hold (that is, if such a $\mu$ does not exist), we can find a large low dimensional subset $V'$. The high level idea is to consider the greedy distribution on $d$-tuples that is sampled as follows: iteratively pick a random unspanned element from $V$ and add it to the spanning set until we cover all of $V$. If this distribution gives low probabilities for many elements of $V$ then we show that it must be due to the fact that these elements lie in some low dimensional subspace. The following definition will be crucial to this argument.

\begin{define}[$(\eta,\tau)$-independent set]
Let $U = (u_1, \ldots,u_n) \in (\R^d)^n$ be a list of $n$ points spanning $\R^d$. We say that $U$ is $(\eta,\tau)$-independent, if there exists a distribution $\mu$ supported on $\cB(U)$, and a set $S \subseteq [n]$ with $|S| \geq (1-\eta)n$ such that for all $j \in S$	
$$\tau \frac{d}{n} \leq \Pr_{I \sim \mu}[ u_j \in I]$$
\end{define}

Since every $I \sim \mu$ has exactly $d$ elements, observe that for every distribution $\mu$, $$E_j[\Pr_{I \sim \mu}[ u_j \in I]] = d/n.$$ Moreover, if the points were in ``general position", i.e. every $d$ of the points were linearly independent, then by taking the distribution $\mu$ to be the uniform distribution on $d$-tuples with distinct elements, we would get a $(0,1)$-independent set.

\begin{lem}\label{lem:notbarthe}
Let $U = (u_1, \ldots,u_n) \in (\R^d)^n$. If $U$ is not $(\eta,\tau)$-independent, then there exists a subspace $W$ of dimension at most $2\tau d$ which contains at least $\eta n/2$ elements of $U$. 
\end{lem}
\begin{proof}
Consider the following distribution $\mu$  supported on $\cB(U)$ that is sampled as follows: For $i$ going from $1$ to $d$, sample $u_i'$ uniformly at random from $U \setminus {\span(u_1', u_2', \ldots u_{i-1}')}$. 
Since $U$  is not $(\eta,\tau)$-independent, there exists a set $T \subset [n]$, with $|T| \geq \eta n$, such that for all $j \in T$	
$$\tau \frac{d}{n} \geq \Pr_{I \sim \mu}[ u_j \in I].$$

For $t \geq 2 \tau d$, uniformly sample $t$ linearly independent vectors $u_1^\ast,\ldots, u_t^\ast$ from $U$ and let $W$ be the subspace they span. Observe that the distribution on  $u_1^\ast,\ldots, u_t^\ast$ is the same as that obtained by taking a sample from $\mu$ and keeping only the first $t$ vectors in the list. Call this distribution $\mu^{(t)}$.
\begin{claim}\label{clm:notbarthe}
For every vector $u \in T$, $\Pr[u \in  W] \geq 1/2$.
\end{claim}

\begin{proof}[Proof of Claim~\ref{clm:notbarthe}]
Let $u \in T$. Let $A$ be the event that $u \in (U \setminus W) \cup \{u_1^\ast,\ldots, u_t^\ast\}$.
Let $p = \Pr[A]$. 
Observe that, as long as the vector $u$ is {\em not} picked, the $i$th vector in the  distribution $\mu^{(t)} \mid A$ is sampled uniformly at random from $(U \setminus {\span(u, u_1^\ast, u_2^\ast, \ldots u_{i-1}^\ast)}) \cup {u}$. 
Therefore, $$\Pr_{I \sim \mu^{(t)}\mid A}[ u \in I] \geq 1-\prod_{i=1}^t (1-1/n-i+1) = t/n \geq 2 \tau d/n.$$
However, $$\Pr_{I \sim \mu^{(t)}\mid A}[ u \in I] = \Pr_{I \sim \mu^{(t)}}[ u \in I]/\Pr[A] \leq \Pr_{I \sim \mu}[ u \in I]/\Pr[A] \leq \tau d/n\Pr[A].$$

Thus $$ p = \Pr[A] \leq 1/2.$$
Hence it follows that $\Pr[u \in  W] \geq 1/2$.

%Fix the notational issue caused by mu being either a distribution on indices or a distribution on vectors. 
\end{proof}

Now the lemma easily follows, since Claim~\ref{clm:notbarthe} implies that the expected number of vectors in $T$ that lie in $W$ is at least $(1/2)|T| \geq \eta n/2$. Thus there exists a fixed subspace $W$ of dimension at most $2\tau d$ which contains at least $\eta n/2$ vectors of $U$. 
\end{proof}

\proof[Proof of Lemma~\ref{lem-reducewell}]

Applying Lemma~\ref{lem:notbarthe} we get that $V$ must be $(\delta/2, 2\beta)$-independent. Otherwise, $V$ would contain a subset $V'$ of size $(\delta/4)n$ and dimension at most $4\beta d$ (contradicting the assumption in the lemma). Hence, there exists a distribution $\mu$  on $\cB(U)$ and a set $S \subset [n]$ with $|S| \geq (1-\delta/2)n$ such that for all $j \in S$	
$$2\beta \frac{d}{n} \leq \Pr_{I \sim \mu}[ j \in I].$$ Let $U = V_S = \{v_i  \mid i \in S\} = (u_1,\ldots,u_{n'})$ with $n' = |S|$.
Lemma~\ref{cor-barthe2} now implies that there  there exists an invertible linear map $M$ so that, denoting $\hat u_j = \frac{Mu_j}{\|Mu_j\|}$, we have for all unit vectors $w \in \R^d$
\begin{equation*}
\sum_{j\in S}  \ip{\hat u_j}{w}^2  \leq \frac{n}{\beta d} 
\end{equation*}
Notice that $U$ is a $(3,\delta/2)$-LCC since the complement of $U$ can intersect at most $\delta n/2$ triples from each matching in $V$. This completes the proof of the lemma. \qed

\subsection{A convenient form of Barthe's theorem}

The proof of Lemma~\ref{lem-reducewell} actually gives a more general result (not mentioning LCCs) that might be of independent interest.

\begin{thm}\label{thm-bartheconvenient}
Let $V = (v_1,\ldots,v_n) \in (\R^d)^n$ with $\dim(V) = d$ be so that any subset $U \subset V$ of size $|U| \geq \alpha n$ has $\dim(U) \geq \beta d$. Then, there exists an invertible linear map $M: \R^d \mapsto \R^d$ and a subset $S \subset V$ of size $|S| \geq (1 - 2\alpha)n$ so that, if we denote by $\hat v = \frac{Mv}{\|Mv\|}$, we have for all unit vectors $w \in \R^d$
$$ \sum_{v \in S} \ip{\hat v}{w}^2 \leq \frac{4n}{\beta d}. $$
\end{thm}
\begin{proof}
The conditions on $V$ and Lemma~\ref{lem:notbarthe} imply that $V$ is $(2\alpha,\beta/2)$-independent. Then, using Lemma~\ref{cor-barthe2}, we get the map $M$ and a set $S$ as required.
\end{proof}

%%%%%%%%%%%%%%%%%%%%%%%%%%%%%%%%%%%%%%%%%%%%%%%%%%%%%%%%%%%%%%%%%%%%%%%%%%%%%
%%%%%%%%%%%%%%%%%%%%%%%%%%%%%%%%%%%%%%%%%%%%%%%%%%%%%%%%%%%%%%%%%%%%%%%%%%%%%%
%%%%%%%%%%%%%%%%%%%%%%%%%%%%%%%%%%%%%%%%%%%%%%%%%%%%%%%%%%%%%%%%%%%%%%%%%%%%%%
%%%%%%%%%%%%%%%%%%%%%%%%%%%%%%%%%%%%%%%%%%%%%%%%%%%%%%%%%%%%%%%%%%%%%%%%%%%%%%
%%%%%%%%%%%%%%%%%%%%%%%%%%%%%%%%%%%%%%%%%%%%%%%%%%%%%%%%%%%%%%%%%%%%%%%%%%%%%%

\section{Reduction to the low triple-multiplicity case}\label{sec-reducelowmult}

In this section we prove a lemma showing that, when analyzing a $(3,\delta)$-LCC $V$ over any field $\F$, it is enough to consider codes in which the matchings $M_v, v \in V$ used in the decoding are such that each triple appears in a small number of matchings (otherwise we can find a large low dimensional subset). 

\begin{define}[Triple multiplicity]
We say that a $(3,\delta)$-LCC $V$ with matchings $M_v, v\in V$ satisfy {\em triple multiplicity} at most $r$  if each triple in each $M_v$ appears in at most  $r$ of the matchings.
\end{define}

\begin{lem}\label{lem-reducelowmult}
Let $\F$ be a field, $n \geq (1/\delta)^{\omega(1)}$ and $\beta >0$ a constant. Let $V = (v_1,\ldots,v_n) \in (\F^d)^n$ be a $(3,\delta)$-LCC with matchings $M_v,v \in V$. Suppose that any subset $V' \subset V$ with $|V'| > (\delta^2/36)n$ satisfies  $\dim(V') > n^{1/2-\beta/4}$. Then, there exists a $(3,\delta/24)$-LCC  $\,U \subset V$ with $|U|  \geq (\delta/4)n$ and matchings $M'_v, v \in U$ so that $U$ (with the matchings $M'_v$) has triple multiplicity at most $n^\beta$ and the matchings $M'_v$ are subsets of the corresponding matchings $M_v$.
\end{lem}

\begin{proof}
	We first reduce to the situation where every element participates in many triples. Unless mentioned otherwise, we will count triples with multiplicity. Let $0<  \gamma = \delta^2/6$ be a real number. Iteratively delete vertices from $V$ that participate in $< \gamma n$ triples (counted with multiplicity), and the triples they participate in.  Let $B \subseteq V$ be be the subset of deleted elements, and let $V' = V \setminus B$. Since each deleted vertex only gets rid of $\gamma n$ triples, the total number of triples which include some vertex of $B$ is at most $\gamma n^2$.  Thus each element in $V'$ participates in at least $\gamma n$ triples, and at least $(\delta-\gamma) n^2 > (2\delta/3) n^2$ of the triples in $V$ are supported entirely in $V'$. Call this set of triples $T'$.

	\begin{claim} $|V'| > 2\delta n$.
	\end{claim}

	\begin{proof}
	 This is because there must be some $v \in V$ with at least  $(2\delta/3)n$ triples in its matching that still survive in $T'$ -- if this was not the case, we would have $|T'| <  (2\delta/3)n^2$. Since the triples in the matching corresponding to $v$ are disjoint, $|V'| \geq 2\delta n$. 
	\end{proof}

	Let $B' \subset V'$ be the subset of points in $V'$ which have less than  $\delta n/2$ of the triples in their matching supported in $V'$. Let $V'' = V'\setminus B'$. 

	\begin{claim}
	 $|V''| \geq  \delta n$, and $V''$ is a $(3,(\delta/6)(n/|V''|))$-LCC. 
	\end{claim}

	\begin{proof}
	There can be at most $\delta n/3$ elements in $V'$ such that $\delta n/2$ triples in their matchings include an element from $B$ -- if there were more than that, then the total number of triples including a element from $B$ would be greater than $\delta n/3\cdot\delta n/2 \geq \delta^2 n^2/6 \geq \gamma n^2$, which is not possible. Thus, at least $|V'| -\delta n/3$ of the elements in $V'$ have a matching of size at least $\delta n/2$ decoding them, lying wholly within $V'$. Thus $|B'| \leq \delta n/3 $. Hence $|V''| \geq |V'| - |B'| \geq |V'|- \delta n/3 > \delta n$. Moreover, for each $v \in V''$, it has a matching of size at least $\delta n/2 - |B'| \geq\delta n/6$ supported in $V''$. Thus $V''$ is a $(3,(\delta/6)(n/|V''|))$-LCC. Let $T''$ be the union of all the triples in the LCC $V''$. 

	\end{proof}

We will call a triple in $T''$ a {\em high multiplicity} triple if it has multiplicity at least  $n^\beta$  in $T''$ (otherwise we will call it a {\em low multiplicity triple}).

\begin{claim}\label{cla-lowmult}
At least  $(1- \delta/24)|V''|$ of the elements in $V''$ have a matching of size $(\delta/12) |V''|$ of low multiplicity triples decoding them.
\end{claim}
\begin{proof}
Suppose the claim does not hold. That is, at least  $(\delta/24) |V''|$ of the elements in $V''$ have at least half  of their matchings (in $T''$) composed of high multiplicity triples. 

	We now delete all the triples of low multiplicity from $T''$. Since there are at least $(\delta^2/288) |V''|^2$ triples (counting multiplicity) of multiplicity  at least $n^{\beta}$ in the LCC  $V''$, by averaging, there exists $v \in V''$ that participates in at least $(\delta^2/288)|V''|$ triples (counted with multiplicity), and each of the triples has multiplicity at least $n^{\beta}$. Observe that since all these triples contain $v$, no two triples are part of a matching corresponding to the same element. 

	By greedily choosing distinct triples containing $v$ of highest multiplicity, one can pick a set $T^\ast$ of distinct triples of size at most $n^{1/2- \beta/2}$ such that together they span at least $n^{1/2+\beta/2}$ distinct elements  of $V''$ (since $n^{1/2+\beta/2} \leq (\delta^2/288)|V''|$, and each triple of multiplicity $n^{\beta}$ spans at least $n^{\beta}$ distinct elements, and distinct triples sharing an element must span distinct elements). 

	Let $L$ be a linear transformation of co-rank at most $3n^{1/2-\beta/2}$ which maps each element participating in a triple of $T^\ast$ to $0$. Since all the elements spanned by the triples of $T^\ast$ also get mapped to $0$, at least $n^{1/2+\beta/2}$ elements of $V''$ get mapped to $0$ under $L$. Let this set be $V^\ast$. Recall that each element of $V'$ (and hence of $V^\ast$) participates in $\gamma n$ triples which together decode $\gamma n$ distinct elements of $V$. 

	Let $S\subset V$ be the subset of all elements whose matching contains at least $(\gamma/6) n^{1/2+\beta/2}$ triples that each contain some element from $V^\ast$. Since the total number of triples containing some element from $V^\ast$ is at least $|V^\ast| \cdot \gamma n/3$, by a simple counting argument we get that $|S| \geq (\gamma/6) n$.  

	\begin{claim}
	$\dim(S) \leq 2n^{1/2-\beta/3}< n^{1/2-\beta/4}$.
	\end{claim}

	\begin{proof}
	If possible let $\dim(S) >  2n^{1/2-\beta/3}$, then $\dim(L(S)) > 2n^{1/2-\beta/3} - 3n^{1/2-\beta/2}> n^{1/2-\beta/3} $. Moreover, since $L$ sends $V^\ast$ to $0$, all triples containing some element of $V^\ast$ now have at most $2$ nonzero elements, and thus the triples can be replaced by {\it pairs}.  Thus $L(V)$ is a $(2,(\gamma/6)n^{-1/2+\beta/2})$-LDC of size $n$, decoding to linearly independent vectors spanning at least $n^{1/2-\beta/3}$ dimensions.  Using Theorem~\ref{thm-2LDC} (lower bound for 2-query LDCs) we  get that $$n \geq 2^{\frac{\gamma/6 n^{\beta/6}}{16} -1}.$$ 

	Since $n \geq (1/\delta)^{\omega(1)}$, $\gamma= \poly(\delta)$ and $\beta = \Omega(1)$, this is a contradiction (for large enough $n$). 

	\end{proof}

Thus, the set $S$ has size at least $(\gamma/6) n = \delta^2n/36$ and  dimension at most $n^{1/2-\beta/4}$, contradicting the assumption in Lemma~\ref{lem-reducelowmult}. This completes the proof of Claim~\ref{cla-lowmult}
	\end{proof}

Applying Claim~\ref{cla-lowmult}, we see that one can delete all triples of multiplicity greater than $n^{\beta}$ and delete at most $\delta|V''|/24$ elements to get a subset $U$ such that each element of $U$ has a matching of $\delta |U|/24 $ triples decoding to it where the triples are supported in $U$. Thus $U$ is a $(3,\delta/24)$-LCC with $|U| \geq \delta n/4 $, and with all triples of multiplicity at most $n^{\beta}$. This completes the proof of Lemma~\ref{lem-reducelowmult}.

\end{proof}

%%%%%%%%%%%%%%%%%%%%%%%%%%%%%%%%%%%%%%%%%%%%%%%%%%%%%%%%%%%%%%%%%%%%%%%%%%%%%%
%%%%%%%%%%%%%%%%%%%%%%%%%%%%%%%%%%%%%%%%%%%%%%%%%%%%%%%%%%%%%%%%%%%%%%%%%%%%%%
%%%%%%%%%%%%%%%%%%%%%%%%%%%%%%%%%%%%%%%%%%%%%%%%%%%%%%%%%%%%%%%%%%%%%%%%%%%%%%

\section{LCCs over $\R$ can be clustered}\label{sec-reducecluster}
In this section we prove the `clustering step' described in the introduction. 

\begin{define}[Clustering]\label{def-clsutering}
Let $S_1,\ldots,S_m \subset [n]$. We say that a triple $\tau \in {[n] \choose 3}$ is {\em clustered} by the family of sets $S_1,\ldots,S_m$ if there exists $i \in [m]$ so that $|\tau \cap S_i| \geq 2$. If $M$ is a multiset of triples, we say that $M$ is clustered by $S_1,\ldots,S_m$ if every triple in $M$ is clustered.
\end{define}

We prove the clustering result as a sequence of three lemmas. First we state the final clustering lemma that will be used later in the proof of our main result.

\begin{lem}[Final clustering]\label{lem-reducecluster}
Let $n> (1/\delta)^{\omega(1)}$ and let $\beta>0$ be a constant. Let $V = (v_1,\ldots,v_n) \in (\R^d)^n$ be a $(3,\delta)$-LCC so that 
every subset $U \subset V$ of size $|U| \geq (\delta^2/288)n$ has dimension  at least $\max\{ 8\delta^6d, n^{1/2 - \beta/4}\}$. Then, there exists a $(3,\hat \delta)$-LCC $\hat V = (\hat v_1,\ldots,\hat v_{\hat n}) \subset V$ of dimension $\hat d \leq d$, size $\hat n \geq (\delta/10)n$ and $\hat \delta \geq \delta^2/4$ and sets $S_1,\ldots,S_m \subset [\hat n]$ so that
\begin{enumerate}
	\item $|S_i| \leq O(\hat n/\hat \delta^6 \hat d)$ for all $i \in [m]$.
	\item $\Omega(\hat \delta^{19} \hat d^3/\hat n^{1+2\beta}) \leq m \leq O(\hat n^{1+2\beta}/\hat \delta^{10}  \hat d).$
	\item If $\hat M_{\hat v}, \hat v \in \hat V$ are the matchings used to decode $\hat V$, then every triple in each $\hat M_{\hat v}$ is clustered by $S_1,\ldots,S_m$.
\end{enumerate}
\end{lem}

We will prove this lemma using the following lemma, which adds conditions on the given code.
\begin{lem}[Intermediate Clustering]\label{lem-clusterintermediate}
Let $n \geq (1/\delta)^{\omega(1)}$ and $\beta>0$ a constant. Let $V =(v_1,\ldots,v_n) \in (\R^d)^n$ be a  $(3,\delta)$-LCC with triple multiplicity at most $n^\beta$ and so that for each  unit vector $w \in \R^d$
	\begin{equation*}
	\sum_{j=1}^n  \ip{v_j}{w}^2  \leq \frac{n}{\delta^6d}.
	\end{equation*}
Let $ t = \frac{n}{\delta^6d}$ and suppose that $d > \frac{10^8 \cdot 200 }{\delta^8}$. Then, there exist  $m$ subsets $S_1,\ldots,S_m \subset V$ such that
	\begin{enumerate}
		\item $|S_i| \leq O(t)$ for all $i \in [m].$
		\item $\Omega(\delta n^{2-\beta}/t^3) \leq m \leq O(t\cdot n^\beta/\delta^4)$.
		\item If $M = \cup_{v \in V} M_v$ is the multiset of all triples in all matchings used to decode $V$, then there are at most $\delta^2 n^2/100$ triples in $M$ that are not clustered by $S_1,\ldots,S_m$
	\end{enumerate}
\end{lem}

To prove the intermediate clustering lemma we first prove a basic clustering lemma.
\begin{lem}[Basic Clustering]\label{lem-clusterbasic}
Let $n,t,\beta,\delta$ and $V \in (\R^d)^n$ be as in Lemma~\ref{lem-clusterintermediate} and let $M$ be the multiset of triples obtained by taking the union of all $M_v, v \in V$. Let $\bar M \subset M$ be of size at least $\delta^2 n^2/100$ and  suppose that $d > \frac{10^8 \cdot 200 }{\delta^8}$. Then there exists a subset $S \subset V$ with $|S| \leq O(t)$ and a subset $T \subset \bar M$ with $|T| \geq \Omega(\delta^4 n^{2-\beta} /t) $ such that each triple in $T$ contains at least two elements from $S$.
\end{lem}

The proofs of the two clustering lemmas (Basic Clustering and Intermediate Clustering), are given below after some preliminaries. First, we show how they are used to prove Lemma~\ref{lem-reducecluster}.

\begin{proof}[Proof of Lemma~\ref{lem-reducecluster}]
	At a high level, the proof follows by first applying Lemma~\ref{lem-reducewell} to get the `well spread' condition on the points in a large sub-LCC $V'$ of $V$. Then, we use Lemma~\ref{lem-reducelowmult} on $V'$ to get a subcode $V''$ with low triple multiplicity (this does not ruin the `well spread' condition by much).  Finally, we apply Lemma~\ref{lem-clusterintermediate} on $V''$ to get clustering for almost all triples. The only reason why one of these steps could fail is if we found a large low dimensional subset in $V$ (which will contradict our assumptions). A final refinement step, using Claim~\ref{cla-refinelarge} shows the existence of a subcode $\hat V$ as required. The details follow.

	\paragraph{Reducing to the well-spread case:} We apply Lemma~\ref{lem-reducewell} on $V$, with $\beta = 2\delta^6$,  to obtain a subset $V'$ of size $n' \geq (1 - \delta/2)n$ so that $V'$ is a $(3,\delta'=\delta/2)$-LCC and so that for each unit vector $w \in \R^d$ we have
	\begin{equation*}
	\sum_{v' \in V'}  \ip{v'}{w}^2  \leq \frac{n}{2\delta^6d}.
	\end{equation*}
	If we cannot apply Lemma~\ref{lem-reducewell}, it means that there is a subset $U$ in $V$ of size $|U| \geq (\delta/4)n$ and dimension at most $8 \delta^6 d$, which would contradict our assumptions.

	\paragraph{Reducing to low  triple multiplicity:}
	We now apply Lemma~\ref{lem-reducelowmult} on the LCC $V'$ to get a $(3,\delta/48)$-LCC $V'' \subset V'$ of size $n'' \geq (\delta/8)n$ and with triple multiplicity at most $(n')^{\beta} \leq (n'')^{2\beta}$. If we cannot apply the lemma, it means that there is a subset $U \subset V'$ of size $|U| \geq (\delta'^2/36)n' \geq (\delta^2/288)n$ and dimension $\dim(U) \leq (n')^{1/2 - \beta/4} \leq n^{1/2 - \beta/4},$ which would contradict our assumptions. Let $d'' = \dim(V'')$ and $\delta'' = \delta^2/2$. We can think of $V''$ as a $(3,\delta'')$-LCC over $\R^{d''}$ in which the `well spread condition' above can be written as
	\begin{equation*}
	\sum_{v'' \in V''}  \ip{v''}{w}^2  \leq \frac{n}{2\delta^6d} \leq \frac{n''}{\delta''^6d''},
	\end{equation*}
	for all unit vectors $w \in \R^{d''}$ (we took $\delta'' = \delta^2/2$ to compensate for the drop in $n''$ in the above inequality). Notice that moving from $\R^d$ to $\R^{d''}$ is not a problem since we can orthogonally project all vectors on the span of $V''$ and maintain all inner products with all unit vectors.

	\paragraph{Clustering:} We  can now apply Lemma~\ref{lem-clusterintermediate} on $V''$ to find sets $S_1,\ldots,S_m$ that cluster all but $(\delta''^2/100)n''^2$ of the triples in the decoding matchings of $V''$. With $|S_i| \leq O(n''/\delta''^6 d'')$ for all $i \in [m]$ and (using $t = n''/\delta''^6 d''$) $$\Omega(\delta''^{19} d''^3/n''^{1+2\beta}) \leq m \leq O(n''^{1+2\beta}/\delta''^{10} \cdot d'').$$ If we cannot apply the lemma it means that $d'' \leq (1/\delta'')^{O(1)}$, which would contradict our assumptions on $V$ (since it would have a subset $V''$ of size $n'' \geq (\delta/8)n$ and dimension $(1/\delta)^{O(1)} < n^{1/4}$).

	\paragraph{Refinement:} To complete the proof, observe that, there are at least $(1 - \delta''/10)n''$ points in $V''$ that have at least half of their matchings clustered by $S_1,\ldots,S_m$. Hence, we can use Claim~\ref{cla-refinelarge} to find a $(3,\hat \delta)$-LCC $\hat V \subset V''$ of size $\hat n \geq (1 - \delta''/10)n'' \geq (\delta/10)n$ with $\hat \delta \geq \delta''/2 \geq \delta^2/4$ so that the sets $S_1,\ldots,S_m$ (restricted to indices in $\hat V$) cluster all the triples in the matchings of $\hat V$. Notice that, since $\hat d = d''$, $\hat \delta = \theta(\delta'')$ and $\hat n = \theta(n'')$, the bounds on the sizes of the sets $S_i$ and on $m$ still hold (the difference in constants will be swallowed by the big `O'). This completes the proof of Lemma~\ref{lem-reducecluster}.
\end{proof}

\subsection{Preliminaries  for the proofs of the clustering lemmas}

We denote by $\|v\|$  the $\ell_2$ norm of a vector $v$. Notice that for two unit vectors $u$ and $v$, $\|u-v\|^2 = 2 -2\ip uv$. We denote the {\em correlation} between two unit vector $v,u$ as $|\ip{v}{u}|$.

Let $V$ be as in Lemma~\ref{lem-clusterintermediate} with matchings $M_v, v\in V$.  The conditions of Lemma~\ref{lem-clusterintermediate} (which we will  assume to hold for the rest of this section) tell us  that for all unit vectors $u \in \R^d$ we have
\begin{equation}
\sum_{j=1}^n  \ip{v_j}{u}^2  \leq \frac{n}{\delta^6 d}=t
\end{equation}

This  gives the following useful claim:
\begin{claim}\label{cla:vectorcorr}
For every unit vector $u \in \R^d$ we have $$|\{v \in V \mid |\ip{v}{u} | \geq \alpha \}| \leq t/\alpha^2.$$ 
\end{claim}

We can also bound the number of points in $V$ that correlate with a given plane:
\begin{claim}\label{cla:planecorr}
Let $P \subset \R^d$ be a two dimensional subspace. We have $$|\{v \in V \mid |\ip{v}{u} | \geq \alpha \textrm{ for some unit vector } u \in P \}| \leq (80/\alpha^3)\cdot t$$ 
\end{claim}
\begin{proof}
Let $K = |\{v \in V \mid |\ip{v}{u} | \geq \alpha \textrm{ for some unit vector } u \in P \}|$. For each such $v \in V$ let $u(v) \in P$ be a unit vector with $|\ip{v}{u(v)}| \geq \alpha$. Now, cover the unit circle in $P$ with at most $20/\alpha$ balls\footnote{We consider balls in $\R^d$} of radius at most  $\alpha/2$. By a pigeon hole argument, one of these balls must contain at least $\alpha K/20$ of the points $u(v)$. Now, the center of this ball must have correlation at least $\alpha/2$ with all the $\alpha K/20$ corresponding $v$'s. Applying Claim~\ref{cla:vectorcorr} we get that $K \leq (80/\alpha^3)t$.
\end{proof}

For every unit vector $u \in \R^d$, let $$\Cor(u) = \{v \in V \mid |\ip uv| \geq 1/10^4\}.$$  For every $v \in V$, let $M_v^{\ast} \subseteq M_v$ be defined as $$ M_v^\ast = \{(v_i, v_j, v_k) \in M_v \mid v_i, v_j, v_k \in V \setminus \Cor(v)\} $$ be the subset of the triples decoding $v$ where each vector in each triple has low correlation  with $v$. Intuitively, such triples must be close to a two dimensional plane and hence `almost' dependent.

The following is an immediate corollary of Claim~\ref{cla:vectorcorr}. 
\begin{claim}\label{cla:sizemvcorr}
For every $v \in V$, $|M_v^{\ast}| \geq |M_v| -  10^8t \geq \delta n - 10^8 t$. 
\end{claim}

Let $M^{\ast}$ be the (multiset) union of all triples in $M_v^{\ast}$ for all $v\in V$. By  Claim~\ref{cla:sizemvcorr}, $M^\ast$ has size at  least $\delta n^2 - 10^8 tn$.

The following  proposition bounds the number of triples in $M^\ast$ containing a fixed pair of vertices.
\begin{prop}\label{prop-tripleupperbound}
For all $i \neq j \in [n]$, there are at most $O(t n^\beta)$ triples (counting multiplicities) in $M^\ast$ containing the pair $(v_i,v_j)$.
\end{prop}
\begin{proof}
We will show a bound of $O(t)$ on the number of {\em distinct} triples containing $(v_i,v_j)$. The $O(t n^\beta)$ bound will then follow by our assumption on the maximum multiplicity of triples in $M$ (and so also in $M^{\ast}$).

Let $P = \span\{v_i,v_j\}$. Consider a triple $(v_i,v_j,v_k)$ containing $v_i,v_j$ and suppose this triple belongs to some matching $M_v^\ast$. Let $\Pi = \span\{v_k,v\}$ and observe that both planes $P$ and $\Pi$ (both are indeed planes since the property of the LCC being regular implies the distinctness of the points in a triple and the point they are used to decode to) are contained in the three dimensional subspace $\span\{v_i,v_j,v_k\}$. Therefore, they must intersect in some unit vector $w \in P \cap \Pi$. Now, since $|\ip{v_k}{v}| \leq 10^{-4}$, a simple calculation shows that $w$ must have correlation at least $1/10$ with either $v_k$ or $v$ (since $w$ belongs to their span and they are close to being orthogonal). To summarize, we have shown that in every triple $(v_i,v_j,v_k) \in M_v$, one of the vectors $v,v_k$ has correlation at least $1/10$ with the plane $P$. Now, the union of $\{v,v_k\}$ as we go over all distinct triples containing $\{v_i,v_j\}$ is at most $O(t)$ by Claim~\ref{cla:planecorr}. If the total number of distinct triples is $r$, then at least $r/2$ of the $v$'s will correlate with $P$ or $r/2$ of the $v_k$'s will correlate with $P$. In either case we see that $r/2 = O(t)$, and hence $r = O(t)$. 

%a point $w \in \span\{v_i,v_j\}$. Let us call triples in which $v$ is correlated with $\span\{v_i,v_j\}$ triples of the first kind and triples in which $v_k$ is correlated with $\span\{v_i,v_j\}$ triples of the second kind  (not to confuse with the notions of type A and type B that we will introduce later).

%Now, suppose in contradiction that there are more than $ct$ distinct triples in $M^\ast$ containing $v_i,v_j$, where $c$ is some large constant to be determined later. Then, either at least $ct/2$ are of the first kind or at least $ct/2$ are of the second kind. If there are $ct/2$ (distinct) triples of the first kind then we have at least $ct/2$ points $v \in V$  that correlate with $\span\{v_i,v_j\}$. By Claim~\ref{cla:planecorr}, we have $ct/2 \leq (80\cdot 10^3)\cdot t$ which is a contradiction if we pick $c$ large enough. Similarly, if there are $ct/2$ triples of the second type, we get too many points $v_k$ (these points will be distinct since the triples are distinct) correlating with $\span\{v_i,v_j\}$.
\end{proof}
 
\begin{define}[Triple types]
We split the triples appearing in $M^{\ast}$ into two {\it Types}.
\begin{itemize}

\item A triple $(v_i,v_j,v_k) \in M^{\ast}$ is defined to be of {\it Type A} if there exists a pair of vertices in the triple, say $(v_i,v_j)$, such that $|\ip{v_i}{v_j}| \geq 9/10$. 

\item A triple $(v_i,v_j,v_k) \in M^{\ast}$ is defined to be of {\it Type B} if $ |\ip{v_i}{v_j}| < 9/10$, $ |\ip{v_j}{v_k}| < 9/10$ and $|\ip{v_i}{v_k}| < 9/10$
\end{itemize}
\end{define}

When we refer to a triple as Type A or B we will implicitly assume that this triple is in $M^\ast$.
% (since triples in $M \setminus M^*$ are not assigned a Type).

We first state and prove three simple propositions that will be useful in the proof of the basic clustering lemma. Below, we will sometimes refer to the elements of $V$  as `vertices'.

\begin{prop}\label{prop:type2corr}
Let $(v_i,v_j,v_k)$ be a triple of Type B then either  $|\ip{v_i}{v_j}| \geq 1/100$ or $|\ip{v_i}{v_k}| \geq 1/100$.
\end{prop}
\begin{proof}
Suppose in contradiction that $\ip{v_i}{v_j} < 1/100$ and $\ip{v_i}{v_k} < 1/100$.

Suppose the triple decodes to the vector $u$ and by an appropriate orthogonal change of basis (which does not change distances or inner products), let us assume that the vectors all  lie in the 3 dimensional space spanned by the unit vectors $e_1,e_2$ and $e_3$. We can also assume that $u = e_1$, $v_i$ is a linear combination of $e_1$ and $e_2$, and $v_j$ and $v_k$ are linear combinations of $e_1,e_2$ and $e_3$.

Since the vectors in the triple are uncorrelated to $u$, their inner product with $e_1$ has absolute value at most $1/10^4$. Since $v_i$ is a unit vector, $\ip{v_i}{e_1}^2 + \ip{v_i}{e_2}^2 = 1$ and hence  $|\ip{v_i}{e_2}| > |\ip{v_i}{e_2}|^2 \geq 1 -1/10^8$. 

Also since $|\ip{v_i}{v_j}| < 1/100$ and $|\ip{v_i}{v_k}| < 1/100$, $|\ip{v_j}{e_2}|< 1/100 \times 10^8/(10^8-1) < 2/100$. Similarly $|\ip{v_k}{e_2}|< 2/100$. Also since $v_j$ is a unit vector, $\ip{v_j}{e_1}^2 + \ip{v_j}{e_2}^2  + \ip{v_j}{e_3}^2 = 1$ and hence $\ip{v_j}{e_3}^2 \geq 1- 1/10^8 -4/10^4$, implying that $|\ip{v_j}{e_3}| \geq \sqrt{99/100}$.  Similarly $|\ip{v_k}{e_3}| \geq \sqrt{99/100}$. Hence $|\ip{v_k}{v_j}| \geq 99/100$, contradicting the property of being Type B. 
\end{proof}

\begin{prop}\label{prop:type2ball}
Suppose $T$ is a set of $m$ {\em distinct} triples of Type B, each sharing the pair $(v_i,v_j)$. Let $S$ be the set of size $m$ containing all the vertices of the triples in $T$ except $v_i$ and $v_j$. Then there is a ball of radius at most $5/10^4$ containing at least $m/10^5$ points of $S$. 
\end{prop}

\begin{proof}
We will first show that every point of $S$ is  close to the subspace through $v_i$ and $v_j$, and then apply a pigeon hole argument. 

Let $v_k \in S$. Then $(v_i,v_j,v_k)$ is a triple of Type B, and in particular the triple is in $M_u^{\ast}$ for some vertex $u$.  

By an appropriate orthogonal change of basis (which does not change distances or inner products), we can assume that the vectors all lie in the 3 dimensional space spanned by the unit vectors $e_1, e_2$ and $e_3$. We can also assume that $v_i = e_1$, $v_j$ is a linear combination of $e_1$ and $e_2$, and $u$ and $v_k$ are linear combinations of $e_1,e_2$ and $e_3$.

Since we have a triple of Type B, $|\ip{v_i}{v_j}| < 9/10$. Thus $|\ip{v_j}{e_1}| < 9/10$. Since $\ip{v_j}{e_1}^2 + \ip{v_j}{e_2}^2 = 1$, this implies that $|\ip{v_j}{e_2}| > 2/5$. Also since $|\ip{u}{v_i}|< 1/10^4$ and $|\ip{u}{v_j}|< 1/10^4$, thus $|\ip{u}{e_1}|< 1/10^4$ and $|\ip{u}{e_2}| < 5/2 \times |\ip{u}{v_j}|< 5/2 \times 1/10^4$. Hence $|\ip{u}{e_3}| = \sqrt{1 -|\ip{u}{e_1}|^2 -  |\ip{u}{e_2}|^2} \geq 1-1/10^7$. 
Since $|\ip{u}{v_k}|< 1/10^4$, we get that $|\ip{v_k}{e_3}| \leq 1/10^4 \times 10^7/(10^7-1) \leq 2/10^4$.
Notice that  $|\ip{v_k}{e_3}|$ is precisely the distance of $v_k$ to the subspace spanned by $v_i$ and $v_j$. 

Now consider the unit circle $C$ in the subspace spanned by $e_1$ and $e_2$. We will show that each element of $S$ is at distance at most  $4/10^4$ from $C$.
To see this, observe that for $v_k \in S$, the projection $\bar v_k$ of $v_k$ onto the subspace spanned by $e_1$ and $e_2$ is of length at least $1-2/10^4$ (by the triangle inequality). Thus $\bar v_k$ is at distance at most $2/10^4$ from $C$ and also at distance at most $2/10^4$ from $v_k$. Thus again by the triangle inequality, the distance between $v_k$ and $C$ is at most $4/10^4$. Now cover $C$ with $10^5$ 2-dimensional discs of radius $1/10^4$. Clearly this can be done. Thus each element $v_k$ in $S$ is at distance at most $5/10^4$ from the center of one of these discs. Thus for one of these discs, there are $m/10^5$ points of $S$ that are at distance at most $5/10^4$ from the center of the disc.

\end{proof}

\begin{prop}\label{prop:degree}
Let $G$ be a edge-weighted $k$-regular  hypergraph on $n$ vertices with $k \geq 2$. Define the degree of a vertex to be the sum of the weights of all hyperedges containing it. Suppose the average degree of a vertex in $G$ is $D$. Then, there exists a vertex induced subgraph $G'$ of $G$ in which every vertex has degree at least $D$.
\end{prop}
\begin{proof}
To obtain $G'$ we iteratively delete vertices whose degree in $G$ is less than $D$. Observe that, after each deletion, the average degree in the hypergraph strictly increases. Thus the process must terminate when all vertices have  degree at least $D$. 
\end{proof}

\subsection{Basic clustering: Proof of Lemma~\ref{lem-clusterbasic}}

We first show that  having many triples of the same type implies that we can find a small set of vertices such that many  of the triples intersect the set in at least two of  their elements.  This will be the main step in the proof of Lemma~\ref{lem-clusterbasic} which is given below. Recall that we have an upper bound of $n^{\beta}$ on the multiplicity of each triple in $M^{\ast}$.

\begin{lem}\label{lem-basiccluster-gamma}
Suppose there is a subset $T$ of $ \gamma n^2$ triples (counting multiplicities) in $M^\ast$ of the same type (either Type A or B), then there is a set $S \subseteq V$ such that $|S| = O(t)$, and at least $\Omega(\gamma^2 n^{2-\beta}/t)$ triples in $T$ intersect $S$ in at least two of their elements.
\end{lem}
\begin{proof} We separate into two cases according to the type of the triples in $T$.
In both cases, we will first refine to the situation where every vertex is incident to many ($\gamma n$) triples.
In both cases we will find a cluster of nearby vertices $V^\ast$, and let $S$ be some kind of neighborhood of  
$V*\ast$ such that every triple which intersects $V^\ast$ will also intersect $S$ in two elements. Since $V^\ast$ will be incident to many triples, we will conclude that many triples intersect $S$ in two elements.
 Moreoever we will ensure that every vertex in $S$ will have some constant {\it correlation} with some fixed carefully chosen vertex $w$. Since every element in $S$ correlates with vertex $w$, Claim~\ref{cla:vectorcorr} implies that $S$ cannot be too large. In the case of Type A triples, the argument is fairly straightforward, whereas in the case of Type B triples the argument is more delicate. 
	
\noindent{\bf Case 1: T has triples of Type A.}

Consider the following weighted graph $H$ on vertex set $V$ in which the edges are all pairs $v_i,v_j$ with $|\ip{v_i}{v_j}| \geq 9/10$ and the weight of an edge $(v_i,v_j)$ is the number of triples in $T$, counting multiplicities, that contain this pair (we can discard edges of weight zero). We define the degree of a vertex $\deg(v)$ as the sum of weights over all edges of $H$ that contain $v$. Since $(1/2)\sum_{v}\deg(v) \geq |T|$ we have that the average degree in $H$ is at least $D = 2|T|/n \geq 2\gamma n$. 

Let $H'$ be a vertex induced subgraph of $H$ in which every vertex has degree at least $D$ (such a subgraph exists by Proposition~\ref{prop:degree}). Let $w$ be any vertex in $H'$ and observe that, by Proposition~\ref{prop-tripleupperbound}, $w$ must have at least $r = \Omega(\gamma n^{1-\beta}/ t)$ distinct neighbors $u_1,\ldots,u_r$ (since the maximal weight of an edge is $O(tn^{\beta})$). Let $V^\ast = \{u_1,\ldots,u_r\}$. We define the set $S$ to contain these vertices $u_1,\ldots,u_r \in V^\ast$ as well as all of their neighbors.

First, we argue that $S$ cannot be too large. To see this, observe that, if $(v_i,v_j)$ is an edge in $H$ then $v_j$ must have $\ell_2$ distance at most $ 1/\sqrt{5}$ from either $v_i$ or $-v_i$. Thus, since all vertices in $S$ are at (graph) distance $\leq 2$ from $w$, we have that they are all contained in the union of two balls of radius $2/\sqrt{5}$ around $w$ and around $-w$. This means that all points in $S$ must have correlation at least $4/6$ with $w$. Using Claim~\ref{cla:vectorcorr} we get that $|S| \leq O(t)$.

To see that there are many triples with two elements in $S$ observe that the sum over all weights of edges touching  $u_1,\ldots,u_r$ is at least $r \cdot \gamma n \geq \Omega(\gamma^2 n^{2-\beta} /t)$ (using the fact that $H'$ has high minimum degree). Since every triple is counted at most $3$ times in this sum we conclude that there are at least $\Omega(\gamma^2 n^{2-\beta} /t)$  triples with a pair in $S$.

\noindent{\bf Case 2: T has triples of Type B.}

Consider the following 3-regular weighted hypergraph $G$: The set of vertices of $G$ is the set $V$. For each triple $(v_i,v_j,v_k) \in T$ we have a hyper-edge in $G$ with weight equal to the multiplicity of that triple in $T$. By Proposition~\ref{prop:degree}, there is a subgraph $G'$ of $G$ such that every vertex of $G'$ is incident to at least $\gamma n$ triples (counting weights) lying within $G'$. 

Pick any vertex $v \in G'$. Let $C_v$ be the multiset $\{v' \in V \mid |\ip{v}{v'}| > 1/100\}$. By Claim~\ref{cla:vectorcorr}, $|C_v| < t \cdot 10^4$. Also, by Proposition~\ref{prop:type2corr}, every triple containing $v$ has another vertex $v'$ such that $|\ip{v}{v'}|> 1/100$. Thus by a simple averaging argument,  it must be that for some $v' \in C_v$, the pair $(v,v')$ participates in at least $\frac{\gamma n}{ |C_v|}$ triples (counting multiplicities). Using the bound on triple multiplicity, we get that there is a {\em set}  $T^\ast$ of at least $\frac{\gamma n}{ |C_v|n^{\beta}}$ distinct triples containing $v$ and $v'$. Thus $|T^\ast| \geq \frac{\gamma n}{|C_v|n^{\beta}} \geq \frac{\gamma n^{1-\beta}}{ 10^4 \cdot t}$ and, by Proposition~\ref{prop:type2ball}, at least $|T^\ast|/10^5$ vertices (of G') lie in a ball of radius $5/10^4$. Call this set of vertices $V^\ast$. Thus what we have so far is a set $V^\ast$ of vertices of $G'$ all lying in a ball of radius $5/10^4$, where $|V^\ast| \geq \frac{\gamma n^{1-\beta}}{10^9 \cdot t}$. 

Recall that every point $v_k$ in $V^\ast$ is incident to at least $\gamma n$ triples lying within $G'$, and, by Proposition~\ref{prop:type2corr}, for each of the triples there exists a vertex $v_k'$ distinct from $v_k$ in that triple such that $|\ip{v_k}{v_k'}| > 1/100$. 

Let $S = \{u \in V \mid \exists w \in V^\ast s.t. |\ip{u}{w}| > 1/100\}$ be the set of all vertices that have correlation at least $1/100$ with some vertex of $V^\ast$. Fix $w\in V^\ast$. Then for any $u \in S$, by definition of $S$, there exists $w' \in V^\ast$ such that $\ip{u}{w'} > 1/100$. Also, since radius of $V^\ast$ is at most $5/10^4$, hence $\|w-w'\| \leq 1/10^3$. Together, these imply that $|\ip{u}{w}| > 1/10^3$. 
Since this holds for all $u \in S$ (and for the same fixed $w$), by Claim~\ref{cla:vectorcorr} we get that $|S| < 10^6t$.

Moreover observe that each triple that intersects $V^\ast$ must intersect $S$ in two elements. Since each tripe in $V^\ast$ is incident to at least $\gamma n$ triples, and each triple is counted at most 3 times, thus there must be at least $\Omega(\gamma n\times \frac{\gamma n^{1-\beta}}{10^9 \cdot t}) =\Omega(\gamma^2 n^{2-\beta} /t)$ triples with a pair in $S$. 
\end{proof}

\begin{proof}[Proof of Lemma~\ref{lem-clusterbasic}]
	Since $d > \frac{200 \cdot 10^8}{\delta^8}$ we have that $$ t = \frac{n}{\delta^6 d} < \frac{\delta^2 n}{200 \cdot 10^8}. $$ Thus, by Claim~\ref{cla:sizemvcorr} we have that for each $v \in V$
	$$ |M_v \setminus M_v^\ast| \leq 10^8 t \leq \delta^2 n /200.$$ So, the set $\bar M^\ast = \bar M \cap M^\ast$ must have size at least $|\bar M| - \delta^2 n^2 / 200 \geq \delta^2 n^2 /200$ triples. At least half of these triples are of the same type (A or B) and so we can apply Lemma~\ref{lem-basiccluster-gamma} with $\gamma=\delta^2/400$ to get the required sets $S$ and triples $T$.
\end{proof}

\subsection{Intermediate clustering: Proof of Lemma~\ref{lem-clusterintermediate}}

We prove Lemma~\ref{lem-clusterintermediate} by iteratively applying Lemma~\ref{lem-clusterbasic} until we have gathered `enough' clustered triples, where we call a triple `clustered' if it has intersection size at least 2 with one of the sets $S_i$. 

We start with $\bar M = M$, which is initially of size $|\bar M| \geq \delta n^2 \geq \delta^2 n^2/100$. Applying Lemma~\ref{lem-clusterbasic} we get  sets $T_1 \subset M$ and $S_1 \subset V$ with $|S_1| \leq O(t)$ and so that all triples in $T_1$ are clustered. We now let $\bar M = M \setminus T_1$ and continue in this manner to generate $S_2,S_3,\ldots,S_m$ and (disjoint) $T_2,T_3,\ldots,T_m$, removing the triples in the $T_i$'s from $\bar M$ as we proceed, until there are at most $\delta^2 n^2/100$ triples in $M$ that are not clustered. 

This only leaves the task of bounding the number of iterations, $m$. The upper bound follows from the fact that the sets $T_i$ are disjoint, each of size at least $\Omega(\delta^4 n^{2-\beta}/t)$ and that $|M| \leq \delta n^2$. The lower bound follows from the observation that, by Proposition~\ref{prop-tripleupperbound}, each $T_i$ can have size at most  $|S_i|^2 \cdot O(t \cdot {n^\beta}) = O(n^\beta t^3)$. Since the union of the $T_i$'s contains at least $\Omega(|M|) \geq \Omega(\delta n^2)$ triples we get that $m \geq \Omega(\delta n^{2-\beta}/t^3)$. This completes the proof of Lemma~\ref{lem-clusterintermediate}.

%%%%%%%%%%%%%%%%%%%%%%%%%%%%%%%%%%%%%%%%%%%%%%%%%%%%%%%%%%%%%%%%%%%%%%%%%%%%%%
%%%%%%%%%%%%%%%%%%%%%%%%%%%%%%%%%%%%%%%%%%%%%%%%%%%%%%%%%%%%%%%%%%%%%%%%%%%%%%
%%%%%%%%%%%%%%%%%%%%%%%%%%%%%%%%%%%%%%%%%%%%%%%%%%%%%%%%%%%%%%%%%%%%%%%%%%%%%%

\section{Clustering implies low dimension}\label{sec-clusterlowdim}

The main result of this section is the following lemma giving a dimension upper bound for LCCs in which the triples are `clustered'. Notice that this lemma works over any field $\F$.

\begin{lem}[Clustering implies low dimension]\label{lem-clusterlowdim}
Let $\F$ be a field, $0 < \barconst < 1/50$, $0 < \beta < \barconst/2$ and suppose $n > (1/\delta)^{\omega(1)}$. Let $V = (v_1,\ldots,v_n ) \in (\F^d)^n$ be a $(3,\delta)$-LCC with  matchings $M_v, v \in V$. Suppose there exists sets $S_1,\ldots,S_m \subset [n]$ with 
\begin{enumerate}
	\item $|S_i| \leq O(  n/  \delta^6   d)$ for all $i \in [m]$.
	\item $\Omega(  \delta^{19}   d^3/  n^{1+\beta}) \leq m \leq O(  n^{1+\beta}/  \delta^{10}    d).$
	\item Every triple in each $  M_{ v}$ is clustered by $S_1,\ldots,S_m$.
\end{enumerate}
Then, there is a subset $V' \subset V$ of size $|V'| \geq (\delta/2)n$ and dimension at most $n^{1/2 - \barconst}$.
\end{lem}

This lemma will be an easy corollary of the following lemma, which shows that there is a small subset in $V$ so that, when projecting this set to zero, the dimension of $V$ drops by a lot.

\begin{lem}[Restriction lemma]\label{lem-randomrestrict}
Let $n,\beta, \barconst$, $V$ and $S_1,\ldots,S_m$ satisfy the conditions of Lemma~\ref{lem-clusterlowdim}. Assume further that the matchings $M_v$ are in regular form (no `2-query' triples). If $d > n^{1/2-\barconst}$ then there exists a subset $U \subset V$ with  $$ |U| \leq n^{1/4+7\barconst}$$ such that, if $\cL : \F^d \mapsto \F^d$ is any linear map with $U \subset \ker(\cL)$ then $\cL(V) = \{ \cL(v)\,|\, v \in V\}$ is contained in a subspace of dimension at most $ n^{10\barconst}$
\end{lem}

We prove the Restriction lemma (Lemma~\ref{lem-randomrestrict}) below, following the short proof of Lemma~\ref{lem-clusterlowdim} from Lemma~\ref{lem-randomrestrict}.

\begin{proof}[Proof of Lemma~\ref{lem-clusterlowdim}]
	Using Claim~\ref{cla-regular} we can reduce to the case that the code $V$ and the matchings $M_v$ are in regular form (that is, there are no `2-query' triples). Indeed, replacing $V$ with the code given in Claim~\ref{cla-regular}  leaves us with a new code (with $n$ and $\delta$ the same up to a constant) satisfying the same clustering requirements (using the same sets $S_1,\ldots,S_m$) and with the same dimension. If we cannot apply Claim~\ref{cla-regular} it is because there is a subset $U \subset V$ of size $(\delta/2)n$ and dimension at most $O( (1/\delta)\log(n)) < n^{1/2 - \barconst}$, in which case the proof is done.
	
	Next, Suppose in contradiction that $d > n^{1/2-\barconst}$ (otherwise we let $V' = V$). Apply Lemma~\ref{lem-randomrestrict} to get a subset $U \subset V$ with $|U| \leq n^{1/4+7\barconst}$, such that, if we send $U$ to zero by a linear map, the dimension of $\span\{V\}$ goes down to at most $n^{10\barconst}$. The existence of such a $U$ implies that $$d = \dim(V) \leq |U|+ n^{10\barconst} \leq n^{1/4+7\barconst} + n^{10\barconst}$$ which gives a contradiction if $\barconst < 1/50$. \qed.
\end{proof}

\subsection{Proof of Lemma~\ref{lem-randomrestrict}}

Using the assumptions $d > n^{1/2 - \barconst}$  we get that for each $i \in [m]$,  $$|S_i| = O(\delta^{-6}n^{1/2+\barconst})$$ and the number of sets, $m$, is between
$$\Omega(\delta^{19} n^{1/2-3\barconst- \beta}) \leq m \leq O(\delta^{-10}n^{1/2+\barconst + \beta}).$$

For each $v \in V$ we know that all $\delta n$ triples in $M_v$  contain two elements in one of the sets $S_1,\ldots,S_m$. Let $P_v$ denote the set of all these pairs. That is, for each $S_i$, add to $P_v$ all the pairs in $S_i$ that are contained in a triple from $M_v$. We fix some arbitrary way to associate each pair in $P_v$ with a {\em single} set $S_i$ (if this pair is in more than one set $S_i$ just pick one arbitrarily). 

The properties of the sets $P_v, v \in V$ are summarized in the following claim.
\begin{claim}\label{cla-Pv}
Each $P_v$ is a matching of at least $\delta n$ pairs, each pair $(u,w) \in P_v$ is associated with a unique $S_i$ such that $u,w \in S_i$ and there exists a triple in $M_v$ containing both $u$ and $w$.
\end{claim}

\paragraph{The distribution $\mu$:} We denote by $\neg(n)$ any function of $n$ that is asymptotically upper bounded by $\exp(-n^{\alpha})$ for some  constant $\alpha >0$. We use the notation $A \sim \mu$ to mean `the random variable $A$ is sampled according to the distribution $\mu$'.

 We now define a distribution $\mu$ on subsets of $V$. To pick a set $A \sim \mu$ we first pick an index $i \in [m]$ uniformly at random and then pick $A \subset S_i$ to contain each element of $S_i$ independently with probability $n^{-1/4+\barconst}$. If $S_i$ happens to be empty, we let $A$ be the empty set. It will be convenient to treat $\mu$ also as a distribution on pairs of the form $(A,i)$ with $A \subset V$ and $i \in [m]$ so, we will sometimes write $(A,i) \sim \mu$ to denote that $i$ is the random index chosen in the sampling process of $A$ and, other times just write $A \sim \mu$.

\begin{claim}\label{cla-sizeA}
Let $A \sim \mu$ then $$ \Pr[|A| \geq n^{1/4+3\barconst}] \leq \neg(n).$$
\end{claim}
\begin{proof}
Conditioning on the choice of the set $S_i$, the expectation of $|A|$ is at most $|S_i|\cdot n^{-1/4+\barconst} \leq O(\delta^6n^{1/4+2\barconst}) < n^{1/4+3\barconst}/100$. Thus, by a Chernoff bound, the probability that the size of $A$ exceeds $n^{1/4+3\barconst}$ is at most $\neg(n)$. Taking a union bound over the $m$ possible choices of $S_i$ the probability is still $\neg(n)$.
\end{proof}

\begin{obs}\label{obs-mu}
	We can  define a new distribution $\mu'$ that samples $A$ according to $\mu$ until it gets a set $A$ of size at most $n^{1/4+3\barconst}$. By the claim, the statistical distance between $\mu$ and $\mu'$ is at most $\neg(n)$. Hence, as long as we can tolerate a $\neg(n)$ error in our probabilities, we can switch between $\mu$ and $\mu'$ as needed.
\end{obs}

\paragraph{The functions $f_{A,i}(v)$: } For each set $A \subset S_i$ we define a partial function $f_{A,i} : V \mapsto V$. The value $f_{A,i}(v)$ is defined as follows: Consider the pairs in $P_v$ that are associated with $S_i$. If one of these pairs is contained in $A$ then $f_{A,i}(v)$ is defined to be the third element of the triple of $M_v$ associated with that pair. More formally, if there is a pair $u,w \in S_i$ so that a triple $(u,w,z)$ is in $ M_v$ then we define $f_{A,i}(v) = z$. If there is more than one such pair, we pick one arbitrarily, for instance the first one in some fixed order. If there is no such pair, we let $f_{A,i}(v) = \bot$ (undefined). 

We use the notation $x \sim y$, with $x,y \in \F^d$, to denote that $x$ is a constant multiple of $y$ and $y$ is a constant multiple of $x$. That is, either they are both zero, or they are both non zero multiples of each other. Notice that the relation $\sim$ is an equivalence relation.

\begin{claim}\label{cla-sim}
Let $i \in [m], A \subset S_i$ and let $f_{A,i}$ be define as above. If $\cL :\F^d \mapsto \F^d$ is any linear map sending $A$ to zero, then $\cL(v) \sim \cL(f_{A,i}(v))$ for all $v$ for which $f_{A,i}(v) \neq \bot$
\end{claim}
\begin{proof}
If $f_{A,i}(v) \neq \bot$ then there is a triple $(x,y,z) \in M_v$ with $x,y \in A$ and $f_{A,i}(v)=z$. Since $v \in \span\{x,y,z\}$ we get that $\cL(v) \in \span\{\cL(x),\cL(y),\cL(z)\} = \span\{\cL(f_{A,i}(v))\}$. Similarly, since we are assuming that $v$ is not in the span of $x,y$ (since the matchings $M_v$ are in  regular form), $z$ is in the span of $v,x,y$ and so $\cL(z) \in \span\{\cL(v)\}$.
\end{proof}

\paragraph{Probability bounds:} The following three claims give bounds on certain probabilities involving the functions $f_{A,i}$, when $(A,i) \sim \mu'$.

\begin{claim}\label{cla-probbot}
Let $(A,i) \sim \mu'$ and let $v \in V$. Then, $\Pr[f_{A,i}(v) \neq \bot] \geq \Omega(\delta^{17}n^{-3\barconst} )$
\end{claim}
\begin{proof}
By Observation~\ref{obs-mu}, it is enough to analyze the probability for the distribution $\mu$. Fixing $v \in V$ we call a set $S_i$ {\em heavy} if it contains at least $n^{1/2-2\barconst}$ pairs from $P_v$ (recall Claim~\ref{cla-Pv}). Since we are choosing each element of $S_{i}$ with probability $n^{-1/4+\barconst}$, the probability to `miss' a single pair from $P_v$ is exactly $(1 - n^{-1/2+2\barconst})$. If $S_{i}$ is heavy, then the probability that $A$ contains at least one of the pairs in $P_v$ is at least (using the fact that $P_v$ is a matching):
\begin{equation}\label{eq-probheavy}
\Pr\left[ P_v \cap {A \choose 2} \neq \emptyset\right] \geq 1 - \left( 1 - n^{-1/2+2\barconst}\right)^{n^{1/2-2\barconst} } \geq 1/2.
\end{equation}

We now bound from below the probability that $S_{i}$ is heavy. Recall that $|P_v| \geq \delta n$ and that $m \leq O(\delta^{-10}n^{1/2+\barconst+\beta})$. Let $m_h +m_\ell = m$ so that $m_h$ is the number of heavy sets $S_i$. Since each $S_i$ can contain at most $|S_i|/2 = O(\delta^{-6}n^{1/2+\barconst})$ disjoint pairs, we have that
\begin{eqnarray*}
	\delta n &\leq&  m_\ell \cdot n^{1/2-2\barconst} + m_h \cdot O(\delta^{-6}n^{1/2+\barconst}) \\
	&\leq& O(\delta^{-10}n^{1-\barconst+\beta}) + m_h \cdot O(\delta^{-6}n^{1/2+\barconst}).
\end{eqnarray*}
This implies (since $\beta < \barconst/2$) that $$ m_h \geq \Omega(\delta^7 n^{1/2-\barconst}).$$ Therefore,
$$\frac{m_h}{m} \geq \Omega\left( \frac{\delta^7 n^{1/2-\barconst}}{\delta^{-{10}}n^{1/2+\barconst+\beta}} \right) = \Omega(\delta^{17}n^{-3\barconst}).$$

Combining the above two bounds, we get that the probability of picking a heavy cluster {\em and} then picking some pair in $P_v$ is at least $\Omega(\delta^{17}n^{-3\barconst})$.
\end{proof}

\begin{claim}\label{cla-probpair}
Let $(A,i) \sim \mu'$. Then, for all $v,z \in V$,
$$ \Pr[ f_{A,i}(v)=z] \leq O(\delta^{-19}n^{-1+6\barconst})$$
\end{claim}
\begin{proof}
By Observation~\ref{obs-mu}, it is enough to analyze the probability for the distribution $\mu$. Suppose $z$ appears in a triple $(u,w,z) \in M_v$ that is associated with $S_{\hat i}$ for some $\hat i \in [m]$ (if there is no such $\hat i$ then the probability in question is equal to zero). By our definition of the functions $f_{A,i}$, it is only possible for $f_{A,i}(v)=z$ to hold if $i = \hat i$ and both $u$ and $w$ are chosen to be in the set $A \subset S_{\hat i}$.  The probability to pick $i = \hat i$ is $1/m \leq O(\delta^{-19}n^{-1/2+3\barconst+\beta})$. Now, conditioned on picking this event, the probability of picking both $u$ and $w$ to be in $A$ is $n^{-1/2+2\barconst}$. Multiplying, and using the bound $\beta < \barconst/2$, we get the required bound.
\end{proof}

\begin{claim}\label{cla-probfAbad}
Let $(A,i) \sim \mu'$ and let $B \subset V$ be a set with $|B| \leq n^{1-10\barconst}$. Then, for every $v \in V$,
$$\Pr[ f_{A,i}(v) \neq \bot \,\,\wedge \,\, f_{A,i}(v) \not\in B] \geq \Omega(\delta^{17}n^{-3\barconst}).$$
\end{claim}
\begin{proof}
Let $p = \Pr[ f_{A,i}(v) \neq \bot  \,\,\wedge \,\,f_{A,i}(v) \not\in B] $. Then, by Claims~\ref{cla-probbot} and \ref{cla-probpair}, we have
\begin{eqnarray*}
1-p &=& \Pr[ f_{A,i}(v) = \bot \,\,\vee \,\, f_{A,i}(v) \in B] \\
&\leq& \Pr[ f_{A,i}(v) = \bot]+  \Pr[ f_{A,i}(v) \in B] \\
&\leq& 1 - \Omega(\delta^{17}n^{-3\barconst}) + |B|\cdot O(\delta^{-19}n^{-1+6\barconst})\\
&\leq& 1 - \Omega(\delta^{17}n^{-3\barconst}) + O(\delta^{-19}n^{-4\barconst}).
\end{eqnarray*}
Rearranging, and using the fact that $n \geq (1/\delta)^{\omega(1)}$, we get that $p \geq \Omega(\delta^{17}n^{-3\barconst})$.
\end{proof}

\paragraph{The set $U$: } To define the set $U$ required in Lemma~\ref{lem-randomrestrict}, we proceed as follows. Let $r$ be an integer to be determined later, and pick $r$ sets $A_1,\ldots,A_r \subset V$ and $r$ indices $i_1,\ldots,i_r \in [m]$  so that each $(A_j,i_j)$ is sampled independently according to the distribution $\mu'$. Let $U = \bigcup_{j=1}^r A_j$. Let $f_1 = f_{A_1,i_1},\ldots,f_r = f_{A_r,i_r}$ be the corresponding (partial) functions on $V$. Our goal is to show that, with probability greater than zero, setting $U$ to zero by a linear map, reduces the dimension of $V$ to $n^{10\barconst}$. 

We begin by defining a sequence of undirected graphs $H_0,H_1,\ldots,H_r$ on vertex set $V$ which will depend on the choice of the sets $A_1,\ldots,A_r$. The first graph $H_0$ is the empty graph (containing no edges). We define $H_j$ inductively by adding to $H_{j-1}$ all edges of the form $(v,f_j(v))$ over all $v \in V$. For $j=1\ldots r$, let $k_j$ denote the number of connected components of $H_j$.

\begin{claim}\label{cla-dimcomp}
If $\cL:\F^d \mapsto\F^d$ is any linear map sending $U$ to zero, then $\span\{\cL(V)\}$ has dimension at most $k_r$.
\end{claim}
\begin{proof}
This is an easy corollary of Claim~\ref{cla-sim}. If $\cL(U)=0$ then, for every edge $(x,y)$ in $H_r$, we have $\cL(x) \sim \cL(y)$. Since the relation $\sim$ is transitive, each connected component is contained in a one dimensional subspace after applying $\cL$.
\end{proof}

Let $k_j'$ denote the number of connected components of $H_j$ of size at most $n^{1-10\barconst}$. Call these the `small' components of $H_j$. The next claim bounds the expectation of $k_j'$.

\begin{claim}\label{cla-expectedcj}
Let $1 \leq j \leq r$. Then,  $$\E[k_j'] \leq k_{j-1}'(1 - \Omega(\delta^{17}n^{-3\barconst})).$$
\end{claim}
\begin{proof}
Let $s= k_{j-1}'$ and let $K_1,\ldots,K_{s}$ be the small components of $H_{j-1}$. Pick  representatives $u_i \in K_i$ in each of the components. For each $i =1 \ldots s$, let $X_i$ be an indicator variable so that $X_i=1$ if $f_j(u_i) \in V \setminus K_i$ (that is, $f_j(u_i)$ is defined and does not belong to $K_i$) and $X_i=0$ otherwise (if either $f_j(u_i)=\bot$ or if it is defined but in $K_i$). By Claim~\ref{cla-probfAbad}, we have that $$ \E[X_i]= \Pr[X_i=1] \geq \Omega(\delta^{17}n^{-3\barconst}).$$

Since having an edge $(u_i,f_j(u_i))$ going from $u_i$ to some vertex outside $K_i$ `merges' $K_i$ with another component,  we have that
$$ k_j' \leq s - \frac{1}{2}\sum_{i=1}^s X_i. $$ 
Taking expectations, and using the above bound on the expectations of the $X_i$'s, we get
$$ \E[k_j'] \leq s(1 - \Omega(\delta^{17}n^{-3\barconst}))$$ as was required. 
\end{proof}

Thus, for each $j =1,2,\ldots,r$ there is a choice of a set $A_j \subset S_{i_j}$ such that $H_j$ has at most $k_{j-1}'(1 - \Omega(\delta^{17}n^{-3\barconst})))$ small components. Taking $ r = n^{4\barconst}$, we get that there is a choice of $U$ for which $H_r$ does not have {\em any} small components. Since the number of large components is at most $n^{10\barconst}$, we get:
\begin{claim}
There is a choice of $U$ for which $H_r$ has at most $n^{10\barconst}$ connected components.
\end{claim}

To conclude, we observe that, since we are using the modified distribution $\mu'$, we have $$|U|  \leq r \cdot n^{1/4+3\barconst} \leq n^{1/4+7\barconst}$$ and, using Claim~\ref{cla-dimcomp}, we have that, setting $U$ to zero by a linear map, reduces the dimension of $V$  to at most $n^{10\barconst}$. This completes the proof of Lemma~\ref{lem-randomrestrict}.

%%%%%%%%%%%%%%%%%%%%%%%%%%%%%%%%%%%%%%%%%%%%%%%%%%%%%%%%%%%%%%%%%%%%%%%%%%%%%%
%%%%%%%%%%%%%%%%%%%%%%%%%%%%%%%%%%%%%%%%%%%%%%%%%%%%%%%%%%%%%%%%%%%%%%%%%%%%%%
%%%%%%%%%%%%%%%%%%%%%%%%%%%%%%%%%%%%%%%%%%%%%%%%%%%%%%%%%%%%%%%%%%%%%%%%%%%%%%

\section{Putting it all together - Proof of Theorem~\ref{thm:lccbound}}\label{sec-mainproof}

We will first prove that any $(3,\delta)$-LCC over $\R$ contains a large subset of small dimension. Later we will iterate this to get a global dimension bound.

\begin{lem}\label{lem-subsetlow}
Suppose $n > (1/\delta)^{\omega(1)}$ and let $0 < \barconst < 1/50$. Let $V =(v_1,\ldots,v_n) \in (\R^d)^n$ be a  $(3,\delta)$-LCC. Then, there exists a subset $U \subset V$ of size at least $$ |U| \geq (\delta^3/300)n$$ and dimension at most $$ \dim(U) \leq \max\{ 8\delta^6d, n^{1/2 - \barconst/16}\}.$$
\end{lem}
\begin{proof}
We will prove the lemma by first applying Lemma~\ref{lem-reducecluster} to show that $V$ has a large sub-LCC $V'$ in which the triples cluster. Then, we will apply Lemma~\ref{lem-clusterlowdim} to show that $V'$ has a large low dimensional sub list. The details follow.

Set  $\beta_1 = \barconst/4$ and apply Lemma~\ref{lem-reducecluster} with $\beta = \beta_1$. To apply the lemma we require that $V$ does not contain a subset $U$ of size $(\delta^2/288)n$ and dimension at most $\max\{ 8\delta^6d, n^{1/2 - \beta_1/4}\}$ = $\max\{ 8\delta^6d, n^{1/2 - \barconst/16}\}$. If this is the case, than our proof is done and there is no need to continue. 

Having applied Lemma~\ref{lem-reducecluster}, we obtain a $(3,\delta')$-LCC $V' \subset V$ with $n' = |V'| \geq (\delta/10)n$, $d' = \dim(V') \leq d$, $\delta' \geq \delta^2/4$ and sets $S_1,\ldots,S_m$ which cluster all the triples in the matchings $M_{v'}, v' \in V'$ used to decode $V'$ so that $$|S_i| \leq O(n'/\delta'^6d')$$ and $$\Omega(\delta'^{19}  d'^3/ n'^{1+2\beta_1}) \leq m \leq O(n'^{1+2\beta_1}/\delta'^{10} d').$$
We now apply Lemma~\ref{lem-clusterlowdim} with $\beta = 2\beta_1 < \barconst/2$ and the same $\barconst$ to  conclude that there exist a subset $V'' \subset V'$ of size $$n'' = |V''| \geq (\delta'/2)n' \geq (\delta^2/8)(\delta/10)n \geq (\delta^3/80)n$$ and dimension 
$$ \dim(V'') \leq n''^{1/2-\barconst} \leq \max\{8\delta^6d, n^{1/2 - \barconst/16}\},$$ as was required.
\end{proof}

We now prove an amplification lemma which uses Lemma~\ref{lem-subsetlow} iteratively. For this lemma we will use the following convenient notations: If $S \subset V$ is a subset of $V$, we denote by $\span_V(S) \subset V$ the subset of elements of $V$ that are spanned by elements of $S$ (we think of all these as lists/multisets).

\begin{lem}[Amplification lemma]\label{lem:iterate}
Suppose $n > (1/\delta)^{\omega(1)}$ and let $0 < \barconst < 1/50$. Let $V =(v_1,\ldots,v_n) \in (\R^d)^n$ be a linear $(3,\delta)$-LCC. 
Suppose $S \subset V$ is such that $\span_V(S) = S$ and $S \neq V$. Then there is
a set $S \subseteq S' \subseteq V$ with $\span_V(S') = S'$ such that
\begin{enumerate}
\item Either $S'=V$ or $|S'| \geq |S| + (\delta^4/400)n$.
\item $\dim(S') \leq \dim(S) + \max\{ \delta^6d, n^{1/2 - \barconst/16}\} $.
\end{enumerate}
\end{lem}

We defer the proof of the lemma to the end of this section and proceed with the proof of  Theorem~\ref{thm:lccbound}.

\begin{proof}[Proof of Theorem~\ref{thm:lccbound}]
Let $V = (v_1,\ldots,v_n) \in (\R^d)^n$ be a linear $(3,\delta)$-LCC. We will prove the theorem with $\barconst = 1/1000$ We now apply
Lemma~\ref{lem:iterate} with $\barconst_1 = 1/51$ iteratively. Start with
$S_1 = \emptyset$ and apply Lemma~\ref{lem:iterate}
repeatedly to obtain sets $S_2,S_3,\ldots,$ such that for all $i$,
$$|S_i| \geq |S_{i-1}| +(\delta^4/400)n $$ and
$$\dim(S_i) \leq \dim(S_{i-1}) + \max\{\delta^6d, n^{1/2- \barconst_1/16}\}.$$
Since the size of $S_i$ cannot grow beyond $n$, the process will
terminate after at most $m=\lfloor 400/\delta^4 \rfloor$
steps, yielding $S_m = V$. We then get that
$$\dim(S_m) = \dim(V) \leq  (400/\delta^4)  \max\{\delta^6d, n^{1/2- \barconst_1/16}\} = \max\{(400\delta^2) d, (400/\delta^4) n^{1/2- \barconst_1/16} \} .$$ 

Without loss of generality, for the proof of the theorem we can assume that $\delta^2 < 1/500$. Thus it must be that $d= \dim(V) \leq (400/ \delta^4) n^{1/2- \barconst_1/16} \leq n^{1/2 - \barconst}$. 
This completes the proof of Theorem~\ref{thm:lccbound}.
\end{proof}

\subsection{Proof of Lemma~\ref{lem:iterate}}\label{sec:proof-iterate}

Observe that for $v \in V\setminus S$, all 3 points of any triple in $M_v$ cannot be in $S$ since $\span_{V}(S) =S$. Thus we may assume that $|S| \leq (1 - \delta)n$, since otherwise each vector in $V \setminus S$ would be spanned by the points of $S$ and we would be done. 

\paragraph{Case 1:} There exists $v \in V\setminus S$ such that $\delta n/4$ of the triples in $M_v$ have two of their points contained in $S$. 
In this case let $S' = \span_{V}(\{v\} \cup S)$. Then  $|S'| \geq |S| + (\delta/4) n$, and $\dim(S') \leq \dim(S) + 1$. 

If Case 1 does not hold than each $v \in V\setminus S$, $M_v$ has $3\delta n/4$ of its triples intersecting $S$ in either one or zero points. Let us call a point $v$ {\em type-zero} if if has at least $3\delta n/8$ of its triples contained in $V \setminus S$ and {\em type-one} otherwise. Notice that, if $v$ is type-one, then it must have at least $3\delta n/8$ of its triples intersecting $S$ in exactly one point. We now separate into two additional cases:

\paragraph{Case 2: There are at most $\delta n/4$ type-one points.}

Let $V' \subset V \setminus S$ be the set of  all type-zero points.
Observe that, since $|S| \leq (1 - \delta)n$, we have $|V'| \geq 3\delta n/4$. Also observe that the vectors in $V'$ form a $(3, \delta/8)$-LCC since each point in $V'$ has at least $3\delta n/8 - \delta n/4 = \delta n/8 \geq (\delta/8)|V'|$ triples in its matching contained in $V'$. Using Lemma~\ref{lem-subsetlow} on $V'$ we conclude that there is a subset $U \subset V'$ of size $$ |U| \geq (\delta^3/300)|V'| \geq (\delta^4/400)n$$ and dimension $$ \dim(U) \leq \max\{ 8(\delta/8)^6d', |V'|^{1/2 - \barconst/16}\} \leq \max\{ \delta^6d, n^{1/2 - \barconst/16}\}. $$ Setting $S' = S \cup U$ we are done.

\paragraph{Case 3: There are at least $\delta n/4$ type-one points.} In this case, there are $\delta n/4$ points $v$ in $V \setminus S$, each having at least $3\delta n/8$ of the triples in $M_v$ intersecting $S$ in exactly one point. Let $A$ be a linear transformation whose kernel equals $\span(S)$.  After applying $A$ to $V \setminus S$ we obtain a $(2, 3\delta/4)$ LDC  decoding the  $\delta n/4$ type-one points. Thus the $\delta n/4$ points (after we apply the mapping $A$ to them) must span at most $\poly(1/\delta) \log n \leq \max\{ \delta^6d, n^{1/2 - \barconst/16}\} $ dimensions by Theorem~\ref{thm-2LDC}. Thus, adding them to $S$ will increase the dimension of its span by at most this number. This completes the proof also in this case.

%%%%%%%%%%%%%%%%%%%%%%%%%%%%%%%%%%%%%%%%%%%%%%%%%%%%%%%%%%%%%%%%%%%%%%%%%%%%%%
%%%%%%%%%%%%%%%%%%%%%%%%%%%%%%%%%%%%%%%%%%%%%%%%%%%%%%%%%%%%%%%%%%%%%%%%%%%%%%
%%%%%%%%%%%%%%%%%%%%%%%%%%%%%%%%%%%%%%%%%%%%%%%%%%%%%%%%%%%%%%%%%%%%%%%%%%%%%%

%%%%%%%%%%%%%%%%%%%%%%%%%%%%%%%%%%%%%%%%%%%%%%%%%%%%%%%%%%%%%%%%%%%%%%%%%%%%%
%%%%%%%%%%%%%%%%%%%%%%%%%%%%%%%%%%%%%%%%%%%%%%%%%%%%%%%%%%%%%%%%%%%%%%%%%%%%%%
%%%%%%%%%%%%%%%%%%%%%%%%%%%%%%%%%%%%%%%%%%%%%%%%%%%%%%%%%%%%%%%%%%%%%%%%%%%%%%
%%%%%%%%%%%%%%%%%%%%%%%%%%%%%%%%%%%%%%%%%%%%%%%%%%%%%%%%%%%%%%%%%%%%%%%%%%%%%%
%%%%%%%%%%%%%%%%%%%%%%%%%%%%%%%%%%%%%%%%%%%%%%%%%%%%%%%%%%%%%%%%%%%%%%%%%%%%%%

%%%%%%%%%%%%%%%%%%%%%%%%%%%%%%%%%%%%%%%%%%%%%%%%%5
\bibliographystyle{alpha}

\bibliography{3lcc}
%%%%%%%%%%%%%%%%%%%%%%%%%%%%%%%%%%%%%%%%%%%%%%%%%%%

\end{document}